\documentclass[12pt]{amsart}
\usepackage[nobysame]{amsrefs}
\usepackage[usenames]{color}
\usepackage{amssymb}
\usepackage{amsmath}
\usepackage{amsthm}
\usepackage{amsfonts}
\usepackage{amscd}
\usepackage{graphicx}
\usepackage{bm}

\usepackage[colorlinks=true,
linkcolor=webgreen,
filecolor=webbrown,
citecolor=webgreen]{hyperref}

\makeatletter
\@namedef{subjclassname@2020}{\textup{2020} Mathematics Subject Classification}
\makeatother
\definecolor{webgreen}{rgb}{0,.5,0}
\definecolor{webbrown}{rgb}{.6,0,0}

\usepackage{color}
\usepackage{fullpage}
\usepackage{float}

\usepackage{graphics}
\usepackage{latexsym}
\usepackage{epsf}
\usepackage{breakurl}

\begin{document}

\theoremstyle{plain}
\newtheorem{theorem}{Theorem}
\newtheorem{corollary}[theorem]{Corollary}
\newtheorem{lemma}[theorem]{Lemma}
\newtheorem{proposition}[theorem]{Proposition}

\theoremstyle{definition}
\newtheorem{definition}[theorem]{Definition}
\newtheorem{example}[theorem]{Example}
\newtheorem{conjecture}[theorem]{Conjecture}

\theoremstyle{remark}
\newtheorem{remark}[theorem]{Remark}

\title{Spacetime symmetries and geometric diffusion}

\author{Marc Basquens} \address{Marc Basquens, Department of Energy Technology, 
Royal Institute of Technology (KTH), 10044 Stockholm, Sweden}
\email{marc5@kth.se}

\author{Antonio Lasanta} \address{Antonio Lasanta, Departamento de \'Algebra, Facultad de Educaci\'on, Econom\'ia y Tecnolog\'ia de Ceuta, Universidad de Granada, Cortadura del Valle, s/n, 51001 Ceuta, Spain. Instituto Carlos I de F\'isica Te\'orica y Computacional, Universidad de Granada, E-18071 Granada, Spain. Nanoparticles Trapping Laboratory, Universidad de Granada, Granada, Spain.}
\email{alasanta@ugr.es}

\author{Emanuel Momp\'o} \address{Emanuel Momp\'o, Departamento de Matem\'atica Aplicada, Grupo de Din\'amica No Lineal, Universidad Pontificia Comillas, C. de Alberto Aguilera 25, 28015 Madrid, Spain. Instituto de Investigaci\'on Tecnol\'ogica (IIT), Universidad Pontificia Comillas, 28015 Madrid, Spain}
\email{egmompo@comillas.edu}

\author{Valle Varo} \address{Valle Varo,  Institute for Research in Technology, Universidad Pontificia Comillas, C. de Sta. Cruz de Marcenado, 26, 28015 Madrid, Spain}
\email{mvvaro@comillas.edu}

\author{Eduardo J S Villase\~nor} \address{Eduardo J S Villase\~nor, Departamento de Matem\'aticas,  Universidad Carlos III de Madrid, Avda.\  de la Universidad 30, 28911 Legan\'es, Spain. Grupo de Teorías de Campos y Física Estadística, Unidad Asociada al IEM-CSIC, Serrano 123, 28006 Madrid, Spain}
\email{ejsanche@math.uc3m.es}

\thanks{This work has been supported by the Spanish Ministerio de Ciencia 
Innovación y Universidades-Agencia Estatal de 
Investigación Grant No. AEI/PID2020–116567GB-C22 and PID2021-128970OA-I00. E. J. S. V. is supported by the Madrid Government (Comunidad de Madrid-Spain) under the Multiannual Agreement with UC3M in the line of Excellence of University Professors (EPUC3M23), and in the context of the V PRICIT (Regional Programme of Research and Technological Innovation). A.L. was also partly supported by FEDER/Junta de Andalucía Consejería de Universidad, Investigación e Innovación, and by the European Union, through Grants No. A-FQM-644-UGR20, and No. C-EXP-251-UGR23. The authors want to thank F. Barbero  for many interesting
discussions and comments. }


\title{Spacetime symmetries and geometric diffusion}
\begin{abstract}

We examine relativistic diffusion through the frame and observer bundles associated with a Lorentzian manifold $(M,g)$. Our focus is on spacetimes with a non-trivial isometry group, and we detail the conditions required to find symmetric solutions of the relativistic diffusion equation. Additionally, we analyze the conservation laws associated with the presence of Killing vector fields on $(M,g)$ and their implications for the expressions of the geodesic spray and the vertical Laplacian on both the frame and the observer bundles. Finally, we present several relevant examples of symmetric spacetimes.

\end{abstract}

\maketitle

\section{Introduction}

Relativistic kinetic theory has been studied almost since the inception of general relativity, and its formulation has been developed by many scientists over the years (see \cites{sarbach2013relativistic, sarbach2014geometry, acuna2022introduction} for a historical overview). The kinetic theory of relativistic gases, initially proposed by Synge in 1934 \cite{synge1934energy}, gained significance after the 1960s due to technological advancements and discoveries such as quasars and the cosmic microwave background radiation. Although most of the underlying geometric ideas had been developed in the 1960s by Berger \cites{berger1965lectures, besse2012manifolds} in the Riemannian context, the work of J. Ehlers in the 1970s \cites{ehlers, ehlers1973survey} was particularly important, providing a coherent and robust mathematical framework where relativistic kinetic theory takes place on appropriate submanifolds of the tangent bundle of a spacetime (see \cite{sarbach2014geometry} for a recent review).

The Fokker--Planck (FP) equation describes diffusion processes, with applications in various domains of physics and engineering. In astrophysics, for instance, it is employed to model the theory of cosmic rays \cite{lasuik2019}. Similarly, Fokker--Planck equations are utilized in plasma physics to analyze the effects of near-miss encounters between ions (heavy particles) and electrons \cite{peeters2008}. 

However, unlike relativistic kinetic theory, the field of relativistic diffusion is still in its early stages. The covariant description poses significant challenges. In particular, there are numerous approaches to deriving a Fokker--Planck equation from a system of stochastic differential equations. Due to these ambiguities, multiple models exist in the literature under the label `relativistic Fokker--Planck equation'; for further details, we refer the readers to \cites{debbasch2007relativistic,chevalier2008relativistic,dunkel2009}.
Besides the theoretical problems, experiments to test the proposed theories are difficult to perform.

In the context of differential stochastic equations on a Riemannian manifold, the orthonormal frame bundle proves to be particularly useful. This framework, developed by \cite{Elworthy1988} and \cite{Hsu2002}, is central to the construction of Brownian motion in Riemannian manifolds. In the context of Minkowski spacetime, the relativistic description of a Brownian process was pioneered by Dudley in 1966 \cite{dudley1966lorentz}. Nevertheless, the generalization of Dudley's work to the framework of general relativity (i.e., to a generic Lorentzian manifold $(M,g)$) had to await the research of J. Franchi and Y. Le Jan in 2005 \cite{franchi_jan}. Their diffusion process, initially defined at the level of pseudo-orthonormal frames $\mathcal{SO}^+(M)$, incorporates Brownian noise only in the vertical directions and projects onto a diffusion process on the pseudo-unit tangent bundle (unit observer bundle $UM$).   The infinitesimal generator of their $\mathcal{SO}^+(M)$-valued Stratonovitch stochastic differential equation decomposes into the sum of the vertical Laplacian and the horizontal vector field generating the geodesic flow. This infinitesimal generator allows us to write a relativistic Fokker--Planck equation. This Fokker--Planck equation was used by Calogero \textit{et al.} to describe diffusion in different cosmological settings \cites{felix2010relativistic,calogero2011kinetic,calogero2012cosmological, felix2013newtonian,calogero2013cosmology,felix2014spatially,alho2015dynamics}.

Building upon these ideas, A. Franchi and Y. Le Jan  \cite{franchi_jan} utilized the bundle $\mathcal{SO}^+(M)$ of direct pseudo-orthonormal frames, with fibers modelled on the special Lorentz group (and having their first element in the positive half of the unit pseudo-sphere in the tangent space) to extend the concept of relativistic diffusion to general Lorentzian manifolds. This approach defines a Stratonovich stochastic differential equation that takes values in the $\mathcal{SO}^+(M)$ group, similar to a Langevin equation.  When projected, this equation naturally generates a diffusion process on the mass shell. Moreover, by following what is done in the Riemannian case, we can induce a pure diffusion on $M$ through a pullback operation, which results in a Fokker--Planck-type equation.  

Recent works related to these methodologies include Serva's study \cite{serva2021brownian}, which diverges from the approaches of Franchi and Le Jan by focusing on massless particles and constructing Lorentz invariant processes. Additionally, the work by Andra and Rosyid \cite{AndraDoniRosyid} investigates relativistic diffusion incorporating both diffusion and friction within the framework of $f(R)$-gravity. Complementing these, the paper by Haba \cite{haba2017thermodynamics}, also within $f(R)$-gravity theory, examines the cosmological implications of relativistic diffusion, providing a comprehensive framework for understanding its impact on the universe's large-scale structure and thermodynamic properties.

In this paper, we show that the perspective provided by the frame bundle $\mathcal{SO}^+(M)$ helps to incorporate spacetime symmetries and define symmetric solutions to the Fokker--Planck equations both in the frame and observer bundle description. The careful way of incorporating symmetries into the equations has been exemplified in the case of the Vlasov\footnote{In this context the Vlasov (Liouville) equation describes the evolution of the one-particle distribution function within relativistic kinetic theory, asserting that particles are conserved in phase space. Fokker-Planck reduces to Vlasov equation when the strength of the diffusion process $\sigma$ vanishes.} (Liouville) equation---as demonstrated in works such as  \cite{sarbach2013relativistic} by Sarbach and Zannias---particularly on the mass shell. Here, we present a generalization that encompasses the previous result as a special case while also enabling the incorporation of diffusion.

The structure of the paper is as follows: After this introduction, in Section \ref{section:om}, we present the geometric structures necessary to describe diffusion processes in the frame bundle of a spacetime. We discuss the infinitesimal generator $\mathcal{L}$ of the Franchi--Le Jan-process, and we prove that, when the isometry group of the spacetime is non-trivial, its symmetry-reduced version is consistent. We also elucidate the connection of $\mathcal{L}$  with the corresponding FP equation (symmetric or not) on the $UM$ bundle. In Section \ref{Sec_Tangent}, following a reasoning similar to that of Elhers in \cites{ehlers, ehlers1973survey}, we present a justification of the FP equations based on the conservation of the average number of particle world lines crossing any Cauchy hypersurface. We also define the particle current density, entropy current, and the energy-momentum tensor, along with some of their most important properties. In Section \ref{sec:examples}, we provide the explicit equations in several relevant spacetimes: flat Friedmann–Lemaître–Robertson–Walker cosmologies, the exterior Schwarzschild spacetime, and the Nariai spacetime. Finally, in Section \ref{sec:concl}, we present our conclusions and outline some future lines of research.

\section{Diffusion in the orthonormal frame bundle}
\label{section:om}

\subsection{Geometric framework}

The orthonormal frame bundle provides the natural arena for understanding stochastic phenomena like Brownian motion on a Riemannian manifold \cite{Hsu2002}.
As discussed by \cites{franchi_jan,franchi2012hyperbolic}, this also holds for Lorentzian manifolds of dimension $1+n$. 

Let $(M,g)$ be an $(1+n)$-dimensional oriented, and time-oriented Lorentzian manifold of signature $(-,+,\cdots,+)$. A frame $u\in  \mathcal{F}_x(M)$ at a point $x \in M$ is a linear isomorphism  $u:\mathbb{R}^{1+n}\rightarrow T_xM$. The set 
\[\mathcal{F}(M)=\bigcup_{x\in M} \mathcal{F}_x(M) \ ,\]
of all frames on $M$ constitutes a principal $GL(1+n, \mathbb{R})$-bundle over $M$. This principal bundle is referred to as the frame bundle of $M$. In terms of local coordinates $x^\mu$ on $M$, a frame $u$ can be written in the form
\[u(e_I)=e^\mu_I\frac{\partial }{\partial x^\mu}\,,\quad \det{e^\mu_I}\neq 0\,,\]
where $e_{I}$, $0\leq I\leq n$, denotes the canonical basis of $\mathbb{R}^{1+n}$. Hence $(x^\mu, e ^\mu_I)$ are local bundle coordinates for $\mathcal{F}(M)$. If we focus on the set of $g$-orthonormal frames instead of all linear frames we can construct the $O(1+n)$-principal bundle of orthonormal frames denoted by $\mathcal{O}(M)$. Furthermore, an additional reduction of $\mathcal{F}(M)$ can be considered by changing the gauge group from $GL(1+n, \mathbb{R})$ to the restricted (proper, orthochronous) Lorentz group $SO^+(1,n)$:
\[\mathcal{SO}^+(M)=\bigcup_{x\in M} \mathcal{SO}^+_x(M),\]
where 
\begin{align*}
\mathcal{SO}^+_x(M)&=\{u\in \mathrm{Hom}^+(\mathbb{R}^{1,n},T_xM)\,:\,   g_x(u(e_I), u(e_J)) =\langle e_I, e_J\rangle_{1,n}\,,  0\leq I, J\leq n\,, \\
   &\hspace{5.3cm} u(e_0)\textrm{ future pointing }\}\ .
\end{align*}
Here and in the following we will assume that $e_I\in \mathbb{R}^{1+n}=\mathbb{R}^{1,n}$ satisfy 
\[
\langle e_I, e_J\rangle_{1,n} = \eta_{_{IJ}} = \left\{ 
\begin{array}{rcl}
-1     & \mathrm{if} &\,  I=J=0\\
   1  & \mathrm{if} & \, I=J=i\neq 0\\
         0 &  &\, \mathrm{otherwise}
\end{array}
\right.,
\]
where $\eta_{IJ}$ represents the components of the Minkowski metric on $\mathbb{R}^{1,n}$.

It is straightforward to show that $\mathcal{SO}^+(M)$ is a manifold of dimension $1+n+\binom{1+n}{2}$. 
A frame $u\in \mathcal{SO}_x^+(M)$ is an isometry from the $(1+n)$-dimensional Minkowski space $(\mathbb{R}^{1,n},\langle \cdot,\cdot\rangle_{1,n})$ to $(T_xM,g_x)$ that preserves orientation and time-orientation. We will denote by  $\pi$ be canonical projection $\pi: \mathcal{SO}^+(M) \rightarrow M$ such that, given $u\in \mathcal{SO}_x^+(M)$, $\pi(u) = x$. Notice also that, by construction, $u(e_I)$ may be interpreted as the $I$-th vector of a basis for the tangent space $T_xM$, effectively serving as the $I$-th vector of the frame.

We will denote by $\Phi$ the canonical right action of $SO^+(1,n)$ on $\mathcal{SO}^+(M)$,
\begin{align*}
    \Phi: \mathcal{SO}^+(M)\times SO^+(1,n)&\rightarrow \mathcal{SO}^+(M), \\
    (u,A)& \mapsto \Phi(u,A)=\Phi_u(A)=\Phi^A(u)=uA,
\end{align*}
where $uA$ denotes the map $(uA)(e_I) := u(Ae_I) = u(e_J A^J_{\, I}) = u(e_J) A^J_{\, I}$. This action enables the introduction of the so called \textit{fundamental vector fields} on $\mathcal{SO}^+(M)$. These vector fields are given by the image of the mapping $\lambda: \mathfrak{so}(1,n) \rightarrow \mathfrak{X}(\mathcal{SO}^+(M))$, defined by
\[
(\lambda X)(u) := T_e\Phi_u(X), \quad X \in \mathfrak{so}(1,n),
\]
highlighting that each fundamental vector field is inherently a vertical field, satisfying \[T\pi(\lambda X) = 0,\] and obeys the relation $T\Phi^A(\lambda X) = \lambda( \mathrm{Ad}_{A^{-1}}X)$.

In aligning our methodology with the results of \cite{franchi_jan}, we endow the Lie algebra $\mathfrak{so}(1,n)$ with the trace form $B$ defined by
\[
B(X,Y) = -\frac{1}{2} \mathrm{tr}(XY) \ .
\]
The $\mathfrak{so}(1,n)$-basis $X_{IJ}$, for $0 \leq I < J \leq n$, defined by
\[
(X_{IJ})^K\,_{L} = \delta^K\,_{I}\eta_{JL} - \delta^K\,_J\eta_{IL},
\]
is $B$-orthonormal
\[
B(X_{I_1J_1}, X_{I_2J_2}) = \eta_{I_1I_2}\eta_{J_1J_2}.
\]
Notice that  $X_{0i}$ and $X_{ij}$ correspond to generators of boosts and rotations, respectively. 
The $\mathrm{ad}$-invariant inner product $B$ allows us to equip $SO^+(1,n)$ with a bi-invariant semi-Riemannian metric, denoted as $(\!\!(\cdot , \cdot)\!\!)$, which is defined via the pullback of $B$ by left-translations, $(\!\!(\cdot , \cdot)\!\!)_A = L^*_AB$. The Levi-Civita connection associated with this metric, satisfies the relation, 
\[D_X Y=\frac{1}{2}[X,Y]\,,\quad \forall X,Y\in \mathfrak{so}(1,n)\,.\]
Hence, in terms of the map $\lambda$, the quadratic Casimir element $\mathcal{C}:=c(B)$ associated with $B$ acts as a second-order differential operator on $C^\infty(\mathcal{SO}^+(M))$ through
\[
  \mathcal{C}f:= c(B)f= \mathrm{div} (D f) = \sum_{0\leq I<J \leq n} \eta_{II}\eta_{JJ} (\!\!(X_{IJ}, D_{X_{IJ}} Df)\!\!) = \sum_{0\leq I<J \leq n} \eta_{II}\eta_{JJ} X_{IJ}(X_{IJ} f)\,.
\]
Following again the notation of reference \cite{franchi_jan}, we set  \[V_i := \lambda(X_{0i})\,,\quad 1<i< n\,, \quad  V_{ij} := \lambda(X_{ij}) \quad  \textrm{ for  }  1 \leq i < j \leq n,\] 
and we will write 
\begin{align*}
    \mathcal{C} =\sum_{0\leq I<J \leq n} \eta_{II}\eta_{JJ} \lambda(X_{IJ})^2 =\sum_{1\leq i < j \leq d} V_{ij}^2 -  \sum_{i=1}^d V_i^2  \, .
\end{align*} 
Notice that $SO^+(1,n)$ is not compact. 
The  non-compactness is responsible for the fact that the bi-invariant metric $(\!\!(\cdot , \cdot)\!\!)$  is not Riemannian, and the operator $\mathcal{C}$ is neither positive nor negative definite. We will see that this fact does not prevent us from defining a diffusion operator through $\mathcal{C}$.

Finally, using the Levi-Civita connection $\nabla$ on $(M,g)$, the tangent space $T\mathcal{SO}^+(M)$ can be decomposed in the form
\[
T\mathcal{SO}^+(M) = \mathcal{V}\oplus \mathcal{H} \ ,
\]
where the fibers $\mathcal{V}_u$ of the vertical bundle $\mathcal{V}$ satisfy $\mathcal{V}_u=\ker T_u\pi$ and, on the other hand, the horizontal bundle 
\[\mathcal{H}_u=\{H_{\mathrm{v}}(u)\,:\, \mathrm{v}\in \mathbb{R}^{1+n}\},\] 
is defined in terms of the horizontal vector fields $H_{\mathrm{v}}$ that are characterized by the following construction: given $\mathrm{v}\in \mathbb{R}^{1+n}$ and $u\in \mathcal{SO}^+(M)$, there is a unique horizontal vector field $H_{\mathrm{v}}$ satisfying  \[\theta(H_{\mathrm{v}}(u))=\mathrm{v}\,,\]
where $\theta=e_I \theta^I$ is the $\mathbb{R}^{n+1}$-valued soldering one-form (canonically) defined  by
\[ T_u\pi (v) = u(\theta(v))=\theta^I(v) u(e_I)\,,\quad v\in T_u\mathcal{SO}^+(M)\,.\]
It is also easy to show that $H_{\alpha_1 \mathrm{v}_1+\alpha_2 \mathrm{v}_2 }=\alpha_1 H_{\mathrm{v}_1}+\alpha_2 H_{\mathrm{v}_2}$ and $T\Phi^A(H_\mathrm{v})=H_{A^{-1}\mathrm{v}}$. This implies that, not being right invariant, the vector fields $H_\mathrm{v}$ are not horizontal lifts of vector fields on $M$. In the case $\mathrm{v}=e_I$, we will use the notation $H_I:=H_{e_I}$ to make contact with \cite{franchi_jan}. 

The fields $H_{\mathrm{v}}$ are often called \textit{basic vector fields} and satisfy the following \cite{Bishop}:
\begin{theorem}[Geodesics]
\label{TGeodesics}
Let $\tilde{\gamma}$ be the horizontal lift through $u \in \pi^{-1}(\gamma(0))$ of a smooth curve $\gamma$ in $M$ to $\mathcal{SO}^+(M)$. Then, $\gamma = \pi \circ \tilde{\gamma}$ is a geodesic if and only if there exists a vector $\mathrm{v} = \mathrm{v}^I e_I \in \mathbb{R}^{1+n}$ such that $\tilde{\gamma}$ is an integral curve of the horizontal vector field $H_{\mathrm{v}}$. This condition is satisfied if and only if
\begin{equation}
    \tilde{\gamma}'(s) = \mathrm{v}^I H_I(\tilde{\gamma}(s)),
\end{equation}
where $\tilde{\gamma}'(s)$ denotes the derivative of $\tilde{\gamma}$ with respect to the parameter $s$.
\end{theorem}

\begin{theorem}[Parallelization of $\mathcal{SO}^+(M)$]
\label{TParallelization}
The vector fields $H_I$,  $V_i$, and $V_{ij}$ provide a parallelization of $\mathcal{SO}^+(M)$; that is, for every $u\in \mathcal{SO}^+(M)$, the tangent vectors $H_I(u)$,  $V_i(u)$, and $V_{ij}(u)$ form a basis of $T_u\mathcal{SO}^+(M)$.
\end{theorem}

In terms of adapted coordinates $(x^\mu,e^\mu_I)$ of the frame bundle $\mathcal{F}(M)$, the vector fields $H_I$ and $\lambda (X_{IJ})$ can be written as
\begin{align*}
    H_I &= e^\alpha_I \frac{\partial}{\partial x^\alpha} - e^\mu_I e_J^\nu \Gamma^\alpha_{\mu \nu} \frac{\partial}{\partial e^\alpha_J}\,,\quad \lambda (X_{IJ})=(\eta_{JK}e^\mu_I-\eta_{IK}e^\mu_J)\frac{\partial}{\partial e^\mu_K} \ ,
\end{align*}
where $\Gamma^\alpha_{\mu \nu} $ denotes the Christoffel symbols (i.e. the connection coefficients of the Levi-Civita connection associated with $g$,  expressed in the coordinate basis $x^\mu$: $\nabla_{\partial_{x^\mu}}\partial_{x^\nu}=\Gamma^\alpha_{\mu \nu}\partial_{x^\alpha}$). The coordinate expression for $\lambda(X_{IJ})$ can be equivalently written as
\begin{align*}
    V_i = e^\mu_0 \frac{\partial}{\partial e^\mu_i} + e^\mu_i \frac{\partial}{\partial e^\mu_0} \ , \quad
    V_{ij} = e^\mu_i \frac{\partial}{\partial e^\mu_j} - e^\mu_j \frac{\partial}{\partial e^\mu_i} \ .
\end{align*}

\subsection{Fokker–Planck equation}

The geometric objects introduced in the previous section allowed J. Franchi and  Y. Le Jan in \cite{franchi_jan} to define an $\mathcal{SO}^+(M)$-valued Stratonovitch stochastic differential equation, 
\[
\mathrm{d}\Psi_s = \left(H_0(\Psi_s)\, \mathrm{d}s+\sigma \sum_{i=1}^n V_i(\Psi_s)\circ \mathrm{d}W^i_s\right) \ ,
\]
for $\Psi=(\Psi_s)\in \mathcal{SO}^+(M)$. Here $W=(W^i_s)$ is a $\mathbb{R}^n$-valued Brownian motion, and $\sigma>0$ is a constant that measures the strength of the diffusion process. This equation describes a natural geometric diffusion process in $\mathcal{SO}^{+}(M)$. The generator of this process is
\[
\mathcal{L}:=H_0 + \frac{\sigma^2}{2} \sum_{i=1}^n V^2_i \ ,
\]
whose adjoint $\mathcal{L}^*$ allows us to write the Fokker--Planck equation
\begin{align}
\label{fp_om1}
    \left( H_0 - \frac{\sigma^2}{2} \sum_{i=1}^n V^2_i\right) F = 0 \ .
\end{align}

In section \ref{Sec_Tangent}, we will justify equation \eqref{fp_om1} generalizing the line of reasoning followed by J. Elhers in reference \cite{ehlers} to derive the Vlasov (Liouville) equation.

\subsection{Symmetries} 
In numerous physically relevant contexts, the spacetime isometry group $\mathrm{Iso}(M,g)$ is not trivial and its Lie algebra is characterized through the (complete) Killing vector fields. In other words, the flow of a  Killing vector field $\xi \in \mathfrak{X}(M)$  is an isometry of $(M,g)$ and, hence
\begin{equation}
    \pounds_\xi g = 0.
\end{equation}
It is well known that the presence of Killing fields leads to conservation laws (see, for example, \cite{o1983semi}):  if $\xi$ is a Killing vector field,  the smooth functions 
\[C^\xi_I: \mathcal{F}(M)\rightarrow \mathbb{R}\,, \quad u\mapsto C^\xi_I(u)=g(u(e_I),\xi),\quad 0\leq I\leq n\,,\] satisfy 
\begin{equation}
    \pounds_{H_I} C^\xi_I =0\,. \label{coserv_SO}
\end{equation}
In particular $\pounds_{H_0} C^\xi_0=0$.

In addition to these conservation laws, as discussed in \cites{KN,mok1979complete}, given any vector field $\xi \in \mathfrak{X}(M)$, there exists a \textit{complete lift} (natural lift) of $\xi$ to $\mathcal{F}(M)$ which we will denote as $\xi^{\mathrm{c}} \in \mathfrak{X}(\mathcal{F}(M))$. The complete lift is characterized by the following three properties:
\begin{enumerate}
    \item $\xi^{\mathrm{c}}$ remains invariant under right translations, i.e., $\Phi^{A*}\xi^{\mathrm{c}} = \xi^{\mathrm{c}}$;
    \item The Lie derivative of the soldering form with respect to $\xi^{\mathrm{c}}$ vanishes, $\pounds_{\xi^{\mathrm{c}}} \theta = 0$;
    \item $\xi^{\mathrm{c}}$ is $\pi$-related to $\xi$, such that $T_u\pi (\xi^{\mathrm{c}}) = \xi(\pi(u))$ for all $u \in \mathcal{F}(M)$.
\end{enumerate}
In bundle coordinates, the complete lift is represented as
\begin{equation}
    \xi = \xi^\mu \frac{\partial}{\partial x^\mu} \mapsto \xi^{\mathrm{c}} = \xi^\mu \frac{\partial}{\partial x^\mu} + e^\nu_I \frac{\partial \xi^\mu}{\partial x^\nu} \frac{\partial}{\partial e^\mu_I},
    \label{eq:complete_lift}
\end{equation}
Moreover, if $\xi$ is a Killing vector field, then $\xi^{\mathrm{c}}$ tangentially aligns with the bundle $\mathcal{SO}^+(M)$. This follows from the fact that
\begin{align*}
    \pounds_{\xi^c} F_{IJ}=0 \ ,
\end{align*}
where $F_{IJ}:\mathcal{F}(M)\rightarrow \mathbb{R}^{(2+n)(1+n)/2},\,u\mapsto  F_{IJ}(u)=g(u(e_I),u(e_J))$.

The previous results allow us to prove the following proposition that relates the complete lift of Killing fields with the fundamental and basic vector fields:
\begin{proposition}
        The complete lift $\xi^{\mathrm{c}}$ of a Killing vector field $\xi$ to $\mathcal{SO}^+(M)$ commutes with both fundamental and basic vector fields, specifically with $V_i$, $V_{ij}$, and $H_I$.
    \label{lemma:commutation}
\end{proposition}
\begin{proof}
Given $\xi \in \mathfrak{X}(\mathcal{M})$ with $\pounds_\xi g = 0$, it suffices to verify that the commutator of $\xi^{\mathrm{c}}$ with each basis vector field $(V_i, V_{ij}, H_I)$ vanishes. Employing local coordinates,
\begin{align*}
        [V_i ,\xi^{\mathrm{c}}] 
        =& \left[ e^\mu_0 \frac{\partial}{\partial e^\mu_i} + e^\mu_i \frac{\partial}{\partial e^\mu_0}, \xi^\alpha \frac{\partial}{\partial x^\alpha} + e^\alpha_K \frac{\partial \xi^\beta}{\partial x^\alpha} \frac{\partial}{\partial e^\beta_K} \right] \\
        =& \frac{\partial \xi^\beta}{\partial x^\alpha} \left( \left[ e^\mu_0 \frac{\partial}{\partial e^\mu_i}, e^\alpha_K  \frac{\partial}{\partial e^\beta_K} \right] + \left[ e^\mu_i \frac{\partial}{\partial e^\mu_0}, e^\alpha_K  \frac{\partial}{\partial e^\beta_K}\right] \right) =0\\
        [V_{ij} ,\xi^{\mathrm{c}}] =& \left[e^\mu_j \frac{\partial}{\partial e^\mu_i} - e^\mu_i \frac{\partial}{\partial e^\mu_j}, \xi^\alpha \frac{\partial}{\partial x^\alpha} + e^\alpha_K \frac{\partial \xi^\beta}{\partial x^\alpha} \frac{\partial}{\partial e^\beta_K}  \right] \\
        =& \frac{\partial \xi^\beta}{\partial x^\alpha} \left(  \left[ e^\mu_j \frac{\partial}{\partial e^\mu_i}, e^\alpha_K \frac{\partial}{\partial e^\beta_K} \right] - \left[ e^\mu_i \frac{\partial}{\partial e^\mu_j}, e^\alpha_K  \frac{\partial}{\partial e^\beta_K}\right] \right)=0 \\
        [H_I, \xi^{\mathrm{c}}] =& \left[ e^\mu_I \frac{\partial}{\partial x^\mu} - e^\gamma_I e_J^\nu \Gamma^\mu_{\gamma \nu} \frac{\partial}{\partial e^\mu_J} , \xi^\alpha \frac{\partial}{\partial x^\alpha} + e^\alpha_K \frac{\partial \xi^\beta}{\partial x^\alpha} \frac{\partial}{\partial e^\beta_K} \right] \\
        =& e^\mu_I e^\nu_J \left( \frac{\partial^2 \xi}{\partial x^\mu \partial x^\nu} + \xi^\beta \frac{\partial \Gamma^\alpha_{\mu\nu}}{\partial x^\beta} + \frac{\partial \xi^\beta}{\partial x^\mu}  \Gamma^\alpha_{\beta\nu} + \frac{\partial \xi^\beta}{\partial x^\nu}  \Gamma^\alpha_{\mu\beta} -\frac{\partial \xi^\alpha}{\partial x^\beta}  \Gamma^\beta_{\mu\nu} \right) \frac{\partial}{\partial e^\alpha_J} \\
        =& e^\mu_I e^\nu_J  \left( \pounds_\xi \Gamma \right)^\alpha_{\mu\nu} \frac{\partial}{\partial e^\alpha_J} =0\,,
    \end{align*}
In the last line we have used that $\pounds_\xi g = 0$ and, thus, $\pounds_\xi \Gamma = 0$.
\end{proof}

Note that in the case of the vertical vectors $V_i, V_{ij}$, we have not used the fact that the vector field $\xi^{\mathrm{c}}$ is the lift of a Killing vector field. Hence, for vertical linear operators, the vanishing of the commutator holds for lifts of arbitrary vector fields on $M$ but in this case  $\xi^{\mathrm{c}}$ is not tangent to $\mathcal{SO}^+(M)$.
\bigskip 
\begin{corollary}
\label{corollary}
The complete lift $\xi^{\mathrm{c}}$ of $\xi$ commutes with the second-order differential operators $\lambda(X_{IJ})^2$; in particular, it commutes with $V^2_{i}$, $V^2_{ij}$, and with the Casimir $\mathcal{C}$.
\end{corollary}
\begin{proof} 
Using the fact that 
\[\pounds_{\xi^{\mathrm{c}}}\pounds_{\lambda(X_{IJ})}-\pounds_{\lambda(X_{IJ})}\pounds_{\xi^{\mathrm{c}}}=\pounds_{[\xi^{\mathrm{c}},\lambda(X_{IJ})]}=0,\]
and taking into account that the second-order operator $\lambda(X_{_{IJ}})^2$ can be written in the form $\lambda(X_{IJ})^2 f=\pounds_{X_{IJ}}(\pounds_{X_{IJ}}f)$, we have
\[
\pounds_{\xi^{\mathrm{c}}}\pounds_{X_{IJ}}(\pounds_{X_{IJ}}f)=\pounds_{X_{IJ}}(\pounds_{X_{IJ}}(\pounds_{\xi^{\mathrm{c}}}f))\,.
\]
\end{proof}

With this in mind, we can finally delineate the conditions that symmetric solutions must fulfill to satisfy  a Fokker--Planck type equation.

\begin{proposition}\label{fpsymm}
The symmetric solutions for the FP diffusion equation in $\mathcal{SO}^+(M)$ are those functions simultaneously satisfying
\begin{align}\label{fp_om_symmetry}
     \left( H_0 - \frac{\sigma^2}{2}\sum_{i=1}^n V^2_i \right) F = 0  \, , \quad         \pounds_{\xi^{\mathrm{c}}} F = 0\,,
\end{align}
for every Killing field $\xi$  of $(M,g)$.
\end{proposition}

\begin{remark}
    The vanishing of the commutator of $\xi^{\mathrm{c}}$ with the basic and fundamental vector fields is crucial for the consistency of the symmetric equations because
\[
\left( H_0 - \frac{\sigma^2}{2}\sum_{i=1}^n V^2_i \right) F = 0 \Rightarrow  \pounds_{\xi^{\mathrm{c}}} \left( H_0 - \frac{\sigma^2}{2}\sum_{i=1}^n V^2_i \right) F = 0,
\]
but this last equation is trivially satisfied when $ \pounds_{\xi^{\mathrm{c}}} F=0$ due to commutativity mentioned above. 
\end{remark}

\begin{remark}  The conclusions of the previous corollary are also satisfied for the equation
\[
\left( H_0 + \frac{\sigma^2}{2} \mathcal{C} \right) F = 0 \,.
\]
As we will see in Lemma \ref{fpsymm_UM} of Section \ref{sec_22}, this equation 
(which is not a diffusion equation) induces the \textit{same} FP equation in the unit-tangent bundle.
\end{remark}

\subsection{Double and timelike unit tangent bundles in Lorentzian manifolds}{\label{sec_22}}

The bundle $\mathcal{SO}^+(M)$ is very useful from the point of view of both writing stochastic differential equations and working with symmetries as we have seen so far. However, from a physical point of view, the frames have too many degrees of freedom, and one would like to work in the configuration space. This is the approach directly taken by \cite{ehlers}, developing the equations with the structures of the physically relevant configurations, corresponding to the (unit mass) observer subbundle $UM$ of the tangent bundle $TM$. The introduction of $UM$ by J.A. Thorpe, as referenced in \cite{Th1} and further discussed in \cite{Beem}, underpins the study of space-time singularities by providing a mathematically rigorous platform. The purpose of this subsection is to present these natural structures and to establish a connection with the structures introduced previously.  

In the following we will need to consider the double tangent bundle, denoted as $TTM$, of a semi-Riemannian manifold $(M,g)$. This bundle can be decomposed into vertical and horizontal sub-bundles, represented as:
\begin{equation}
    TTM = \mathcal{V}TM \oplus \mathcal{H}TM.
\end{equation}
For any tangent vector $v \in TM$, the vertical subspace $\mathcal{V}_vTM$ and the horizontal subspace $\mathcal{H}_vTM$ at $v$ are naturally isomorphic to the tangent space at the base point $\mathrm{p}(v)$ of $v$, with $\mathrm{p}: TM \to M$ being the canonical projection.

Given a vector $X \in T_{\mathrm{p}(v)}M$, one can select a curve $\gamma(t)$ in $M$ where $\gamma'(0) = X$. Utilizing the Levi-Civita connection, we construct a curve $c$ in $TM$ satisfying $\gamma = \mathrm{p} \circ c$ and $c(0) = v$, with the condition $\nabla_{\gamma'}c = 0$. Through this construction, $v$ is mapped to $c'(0)$, a process referred to as the horizontal lift of $X$. In local coordinates, this horizontal lift is expressed as:
\[
T_{\mathrm{p}(v)}M\ni X= X^\mu \frac{\partial}{\partial x^\mu} \mapsto X^{\mathcal{H}} =X^\mu \frac{\partial}{\partial x^\mu}-\Gamma^\mu_{\alpha\beta}v^\alpha X^\beta \frac{\partial}{\partial v^\mu}\in \mathcal{H}_vM\,.
\]
where $\Gamma^\mu_{\alpha\beta}$ are the Christoffel symbols of the Levi-Civita connection. Therefore, the identification of vectors in the horizontal subspace $\mathcal{H}_vTM$ with vectors in $T_{\mathrm{p}(v)}M$ is facilitated by the tangent map of the projection $\mathrm{p}: TM \to M$, restricted to $\mathcal{H}TM$. 

On the other hand, concerning the vertical component of a vector in the double tangent bundle $TTM$, for any vector $v$ in the tangent bundle $TM$, the space $T_{\mathrm{p}(v)}M$ is a vector space. This fact establishes a natural isomorphism $\imath_v: T_v T_{\mathrm{p}(v)}M = \mathcal{V}_vTM \rightarrow T_{\mathrm{p}(v)}M$. In local coordinates, this isomorphism takes the form
\[
\imath_v (V^\mu\frac{\partial}{\partial v^\mu}) = V^\mu\frac{\partial}{\partial x^\mu}\,.
\]
The reverse process is facilitated by the so-called vertical lift, which for any $X \in T_{\mathrm{p}(v)}M$, is given by
\[
T_{\mathrm{p}(v)}M \ni X = X^\mu \frac{\partial}{\partial x^\mu} \mapsto X^{\mathcal{V}} = X^\mu \frac{\partial}{\partial v^\mu} \in T_v T_{\mathrm{p}(v)}M\,.
\]
To decompose any vector $W \in TTM$ into its vertical and horizontal components, \[W = W^{\mathrm{ver}} + W^{\mathrm{hor}},\] an additional construct, the connection map $K$ associated with the Levi-Civita connection, is introduced, where $K: TTM \rightarrow \mathcal{V}TM$. The decomposition in coordinates is described as
\[
W^{\mathrm{ver}} = K_v(W)\,, \quad W^{\mathrm{hor}} = W - K_v(W)\,,
\]
and specifically,
\[
K_v\left(X^\mu \frac{\partial}{\partial x^\mu} + V^\mu \frac{\partial}{\partial v^\mu}\right) = \left(V^\mu + \Gamma^\mu_{\alpha\beta} v^\alpha X^\beta\right) \frac{\partial}{\partial v^\mu}\,, \quad W \in T_v TM\,.
\]

These structures enable the definition of a natural horizontal vector field $L  \in \mathfrak{X}(TM)$, mapping $v \mapsto L (v) \in T_v TM$ and determined by the requirements
\[
T\mathrm{p}(L (v)) = v\,, \quad K_v(L (v)) = 0\,, \quad \forall v \in TM\,,
\]
or, equivalently, $L (v)^{\mathrm{hor}} = v^{\mathcal{H}}$, $L (v)^{\mathrm{ver}} = 0$. This vector field, $L $, known as the \textit{geodesic spray}, is characterized in local coordinates by
\[
L(v)  = v^\mu \frac{\partial}{\partial x^\mu}(v) - \Gamma^\mu_{\alpha\beta} v^\alpha v^\beta \frac{\partial}{\partial v^\mu}(v) \,.
\]
In the following, we will also use the canonically vertical vector field $A$ in $TM$ (the \textit{Liouville vector field}) that is,  the generator of the dilations \[a^t:TM\rightarrow TM\,,\quad v\mapsto a^t(v)=t v\,, \quad t\in \mathbb{R}\setminus \{0\}.\]
In bundle coordinates
\[
A(v)=v^\mu\frac{\partial}{\partial v^\mu}(v)\,.
\]

These canonical isomorphisms can be also used to define the Sasaki metric $\textsf{g}$ on $TM$ associated with the metric $g$ on $M$. The Sasaki metric is formulated by designating the vertical bundle $\mathcal{V}_vTM$ and the horizontal bundle $\mathcal{H}_vTM$ as orthogonal components. Mathematically, this is expressed as
\begin{equation}
    \textsf{g}(V,W) = g(T_v\mathrm{p}(V), T_v\mathrm{p}(W)) + g(\imath_vK_v(V), \imath_vK_v(W))\,, \quad V,W \in T_vTM\,.
    \label{sasaki1}
\end{equation}
In local coordinates, the Sasaki metric can be written in detail  as
\begin{equation}
\textsf{g} = g_{\mu\nu} \mathrm{d}x^\mu \otimes\mathrm{d}x^\nu + g_{\mu\nu}(\mathrm{d}v^\mu + \Gamma^\mu_{\alpha\beta} v^\alpha \mathrm{d}x^\beta)\otimes(\mathrm{d}v^\mu + \Gamma^\mu_{\gamma\delta} v^\gamma \mathrm{d}x^\delta)\,.
    \label{sasaki2}
\end{equation}
It is important to note that if the signature of $g$ is $(1,n)$, then, correspondingly, the signature of $\textsf{g}$ is $(2,2n)$, reflecting the doubled dimensionality and the preservation of the manifold's metric properties within its tangent bundle.

As we will see, the pullback of $\textsf{g}$ to the \textit{unit future observer bundle}, 
\begin{align*}
    UM = U^{1}M = \{ v \in TM \mid g(v,v) = -1, v \text{ is future directed} \}\subset TM\,,
\end{align*}
is important for the analysis of the FP equation. $UM$ is a codimension one submanifold of $TM$ and the canonical projection $\tau: UM \rightarrow M$ equips it as a subbundle of $TM$. 

One notable property of $UM$ is that any tangent vector within $\mathcal{H}TM$ also belongs to the tangent space of $UM$ ($\mathcal{H}_vTM \subset T_vUM$). This inclusion is substantiated by the behavior of the function $E: TM \rightarrow \mathbb{R}$, defined by $v \mapsto E(v) = g(v,v)/2$, which maintains a constant value for all vectors in $\mathcal{H}TM$, as indicated by $\pounds_W E= \mathrm{d}E(W) = 0$. This feature ensures that the geodesic flow, generated by the vector field $L $, preserves the structure of $UM$, hence maintaining its integrity under dynamical evolution.

The Sasaki metric $\textsf{g}$ on $TM$ induces a Lorentzian metric $\mathrm{g}$ of signature $(1,2n)$ on $UM$. This metric equips the $n$-dimensional submanifolds $\mathbb{H}_v = \tau^{-1}(\tau(v)) \subset UM$ with a Riemannian (hyperbolic) metric, where each $\mathbb{H}_v$ represents a fiber of the projection $\tau: UM \rightarrow M$ through $v$. The Riemannian metric, characterized by its positive signature and hyperbolic nature, lays the groundwork for defining the vertical Laplacian $\Delta^{\!\mathrm{ver}}$, a second-order differential operator acting on $C^\infty(UM)$. The vertical Laplacian is expressed as
\[
(\Delta^{\!\mathrm{ver}}f)(v) = (\Delta^{\mathbb{H}_v} f\upharpoonright\mathbb{H}_v)(v)\,,
\]
where $f\upharpoonright\mathbb{H}_v$ denotes the restriction of $f$ to $\mathbb{H}_v$, and $\Delta^{\mathbb{H}_v}$ is the Laplace-Beltrami operator associated with the metric induced on $\mathbb{H}_v$ by $(UM, \mathrm{g})$. The local coordinate representation of $\Delta^{\!\mathrm{ver}}$ highlights its dependency on the geometric structure of $UM$ and the dynamics of space-time, as given by
\[
\Delta^{\!\mathrm{ver}} f = \left( n \, v^i \frac{\partial}{\partial v^i} + \left( v^i v^j + g^{ij} \right) \frac{\partial^2}{\partial v^i \partial v^j} \right) f \,.
\]

We already have all the ingredients to establish the form of the Fokker--Planck equation on $UM$ and its relationship with the one introduced in \(\mathcal{SO}^+(M)\). This is achieved through the following:
 
\begin{lemma}
\label{fpsymm_UM}
Let $\pi_0 : \mathcal{SO}^+(M)\rightarrow UM\,,   u \mapsto \ u(e_0)$ and $f \in \mathcal{C}^\infty\left( UM\right)$, then 
\begin{itemize}
\item[(1)] $T\pi_0(V_{ij}) =0$ and $ V_{ij} \left( \pi_0^*f \right) = 0$,
\item[(2)] $T\pi_0(H_0) =L $ and $H_0 \left( \pi_0^* f\right) = \pi_0^* \left( L  f \right)$,
 \item[(3)]  $\displaystyle  \mathcal{C} \left( \pi_0^*f \right) = -\sum_{i=1}^n V_i^2  \left( \pi_0^*f \right)= -\pi_0^*\left( \Delta^{\!\mathrm{ver}} f \right) $
\end{itemize}
where $\Delta^{\!\mathrm{ver}}$ is the vertical Laplacian and $L $ is the geodesic spray of $UM$. 
\end{lemma}
\begin{proof}
Properties (1) and (2) follow almost directly: 
\begin{align*}
  V_{ij}(\pi_0^*f)&=\mathrm{d}(\pi_0^*f)(V_{ij})=\pi_0^*(\mathrm{d}f)(V_{ij})=(\mathrm{d}f)(T\pi_0V_{ij})\circ \pi_0=0\,,\\ 
  H_{0}(\pi_0^*f)&=\mathrm{d}(\pi_0^*f)(H_0)=\pi_0^*(\mathrm{d}f)(H_0)=(\mathrm{d}f)(T\pi_0H_0)\circ \pi_0=
(\mathrm{d}F)(L )\circ \pi_0=\pi_0^*(L f)\,,
\end{align*}
where we have used that  $T\pi_0V_{ij}=0$ and $T\pi_0H_0=L$.

Property (3) can be derived in local coordinates for any function $F\in C^\infty(\mathcal{SO}^+(M))$, 
\begin{align*}
     \sum_{i=1}^n V_i^2 F= \sum_{i=1}^n \left( e^j_i e^k_i \frac{\partial}{\partial e_0^j} \frac{\partial}{\partial e_0^k} +e^j_i \frac{\partial}{\partial e_i^j} + n e_0^j \frac{\partial}{\partial e_0^j} + 2 e_0^j e_i^k \frac{\partial}{\partial e_0^k} \frac{\partial}{\partial e_i^j} + e_0^j e_0^k \frac{\partial}{\partial e_i^j} \frac{\partial}{\partial e_i^k} \right) F\ . 
\end{align*}
Hence, when $F=\pi_0^*f$, the derivatives with respect to $e_i^j$ vanish, and, identifying $e_0^i = v^i$, 
\begin{align*}
   \mathcal{C} \left( \pi_0^*f \right) = -\sum_{i=1}^n V^2_i \left( \pi_0^*f \right)= -\left(  n  \  v^i \frac{\partial}{\partial v^i} + \left( v^i v^j + g^{ij} \right) \frac{\partial}{\partial v^i} \frac{\partial}{\partial v^j}  \right) f .  
\end{align*}
Therefore
\begin{align*}
    \mathcal{C} \left( \pi_0^*f \right) = -\pi_0^*\left( \Delta^{\textrm{ver}} f \right)  \ .
\end{align*}
\end{proof}

Similar to the framework in $\mathcal{SO}^+(M)$, $TM$ supports the \textit{complete lift} of vector fields from $M$. A differential 1-form $\omega \in \Omega^1(M)$, may be viewed as a scalar function $F^\omega: TM \rightarrow \mathbb{R}$ with $v \mapsto F^\omega(v) = \omega(v)$. The complete lift $\xi^{\mathrm{ct}} \in \mathfrak{X}(TM)$ of a vector field $\xi \in \mathfrak{X}(M)$ is uniquely defined by the property
\[
\pounds_{\xi^{\mathrm{ct}}} F^\omega = F^{\pounds_\xi \omega}, \quad \forall \omega \in \Omega^1(M)\,.
\]
In local bundle coordinates $(x^\mu, v^\mu)$ of $TM$, this lift is expressed as
\[
\xi = \xi^\mu \frac{\partial}{\partial x^\mu} \mapsto \xi^{\mathrm{ct}} = \xi^\mu \frac{\partial}{\partial x^\mu} + v^\nu \frac{\partial \xi^\mu}{\partial x^\nu} \frac{\partial}{\partial v^\mu}\,.
\]
When $\xi$ is a Killing field, $\xi^{\mathrm{ct}}$ is tangent to $UM$ and corresponds to the complete lift $\xi^{\mathrm{c}}$ to $\mathcal{SO}^+(M)$ via
\[
T\pi_0(\xi^{\mathrm{c}}) = \xi^{\mathrm{ct}}\,.
\]
This relationship establishes an equivalence with the symmetric equations in $\mathcal{SO}^+(M)$, articulated through the following corollary:
\begin{corollary}\label{corolarioFP}
The diffusion equation in $UM$, represented as the Fokker--Planck equation 
\[
    \left( L  - \frac{\sigma^2}{2}\Delta^{\!\mathrm{ver}} \right) f = 0\,,
\]
is the equation obtained from the diffusion equation \eqref{fp_om1} on $\mathcal{SO}^+(M)$ just by imposing that $F=\pi_0^*f$. Similarly, 
in the presence of spacetime symmetries,  the equations
\[
 \left( L  - \frac{\sigma^2}{2}\Delta^{\!\mathrm{ver}} \right) f = 0 \ , \quad \pounds_{\xi^{\mathrm{ct}}} f = 0 \ ,
\]
are derived from equation \eqref{fp_om_symmetry} on Proposition \ref{fpsymm}.
\end{corollary}
 
Finally, notice that if $\xi$ is a Killing vector field, the function \[C^\xi:TM\rightarrow \mathbb{R}\,,\quad v\mapsto C^\xi(v)=g(v,\xi) \ ,\] satisfies
\begin{equation}
    \pounds_L C^\xi =0\,. \label{conservado}
\end{equation}
These conserved quantities are analogous to \eqref{coserv_SO} for $I=0$.

\section{Relativistic diffusion on the tangent bundle}{\label{Sec_Tangent}}

The development of relativistic kinetic theory, pivotal for understanding relativistic thermodynamics, is grounded in the foundational work by J. Ehlers \cite{ehlers}  in the early 1970s and further explored by others such as Sarbach and Franchi \cites{sarbach2013relativistic, franchi2012hyperbolic, sarbach2014geometry}. This theory extends the classical kinetic theory of gases to a relativistic framework, thereby providing a more comprehensive model of matter that incorporates its particle nature. Central to this theory is the concept of the \textit{one-particle distribution function}, which, analogous to its non-relativistic counterpart, quantifies the expected particle density within a defined volume in the phase space of a single particle.

The objective of this section is to generalize Ehlers' results to allow for the presence of diffusion as well as to try to justify, in an approach closer to physical applications, the FP equation presented in Corollary \ref{corolarioFP}. In the previous sections we focused on the $m=1$ case but in this section we will allow for any value $m>0$ for the mass of the particles.

\subsection{The observer bundle}

To define the \textit{one-particle phase space} for massive particles with arbitrary masses in a $(1+n)$-dimensional, oriented and time-oriented space-time $(M,g)$, we define it as the smooth $2(n+1)$-dimensional submanifold (open subset) of the tangent bundle $TM$:
\[
\mathcal{P} := \left\{ v \in TM : g(v,v) < 0, \, v \text{ is future directed} \right\},
\]
where $\mathcal{P}$ admits a foliation into mass-shell bundles,
\[
\mathcal{P} = \bigcup_{m > 0} U^{m}M,
\]
with 
\[
U^{m}M := \left\{ v \in TM : g(v,v) = -m^2, \, v \text{ is future directed} \right\},
\]
being a $(2n+1)$-dimensional smooth fiber bundle over $M$ for any $m>0$. In particular, the case $m=1$ is just the unit observer bundle presented in the previous section. Each fiber,
\[
U^{m}_xM = \left\{ v \in T_xM : g_x(v,v) = -m^2, \, v \text{ is future directed} \right\},
\]
is isometric to the $n$-dimensional hyperbolic Riemannian space $\mathbb{H}^n(m)$ of constant negative curvature $-1/m^2$, corresponding to the $n$-dimensional (future-pointing) $m$-mass hyperboloid in Minkowski space-time. When convenient, we will write $\mathbb{H}_x=U^{m}_xM$. 

It is straightforward to prove that the Liouville vector field $A$ and the geodesic spray $L$ are orthogonal and tangent  to $U^{m}M$, respectively.  Moreover, the geodesic spray $L$ of $TM$ induces fields on both the open submanifold $\mathcal{P}\subset TM$ and in $U^{m}M$. The notation $L _m$ is used when referring to the geodesic spray as a vector field on $U^{m}M$.

The one-particle phase space, along with the mass-shell bundles and their fibers, inherit significant geometric properties from the tangent bundle $(TM,\textsf{g})$, where $\textsf{g}$ denotes the Sasaki metric defined in equation \eqref{sasaki1}. As regular, orientable submanifolds of $TM$, the pull-back of the Sasaki metric $\textsf{g}$ via the canonical inclusion furnishes $\mathcal{P}$, $U^{m}M$, and $U^{m}_xM$ with semi-Riemannian metrics of signatures $(2,2n)$, $(1,2n)$, and $(0,n)$ respectively. The volume forms corresponding to these metrics are denoted as
\[
\mathrm{vol}_{\textsf{g}}, \quad \mathrm{vol}_{m}, \quad \text{and} \quad \sigma_x^{m}\,.
\]

The volume $\sigma_x^{m}$ appears in the definition of many physical quantities, so it is convenient to keep the following result in mind: 
\begin{lemma}
The volume form $\sigma_x^{m}$, evaluated at $v\in U^{m}_xM$, satisfies 
\[
\sigma_x^{m}(v) = -m \frac{w \lrcorner \mathrm{vol}_{g_x}(v)}{g_x(w,v)}\,,
\]
where $\mathrm{vol}_{g_x}$ is the volume form induced by $g_x$ in $T_xM$ (and $\mathrm{vol}_{g_x}(v)$ the volume of $T_vT_xM$),  $w\in T_vT_xM$ is any vector such that $g_x(w,v)\neq 0$. In particular, using local coordinates and choosing $w = \partial/\partial v^0$, we have
\[
\sigma_x^{m}(v) = -m \frac{\sqrt{-\det g_x}}{g_{0\mu}(x) v^\mu} \,\mathrm{d} v^1\wedge \dots\wedge \mathrm{d} v^n \,,
\]
where $\det g_x = \det (g_{\mu\nu}(x))$ and $v^0$ is given in terms of $(v^1,\dots,v^n)$ by $g_{\mu\nu}(x)v^\mu v^\nu=-m^2$.
\end{lemma}

\begin{proof}
Notice that $T_xM$ is a vector space so $T_vT_xM$ can be identified with $T_xM$. Given $v\in U^{m}_{x}M$, any basis $(b_1,\dots,b_n)$ of $T_vU^{m}_{x}M = v^{\perp_{g_x}}\subset T_vT_xM$, 
\[\mathrm{vol}_{g_x}(w,b_1,\dots,b_n) \ , \]
vanishes when $w\in\mathrm{span}(b_1,\dots,b_n)$. Hence, there exists a volume-form $\sigma_x^m$ on $U^{m}_xM$ such that
\[\mathrm{vol}_{g_x}(w,b_1,\dots,b_n) = -m^{-1} g_{x}(w,v) \sigma_x^m(b_1,\dots,b_n)\,,\quad v\in T_vU^{m}_xM\,.\]
The form
\[
\sigma_x^m(b_1,\dots,b_n) = -m \frac{\mathrm{vol}_{g_x}(w,b_1,\dots,b_n)}{g_{x}(w,v)} \ ,
\]
is independent of the choice of $w$ (as long as $g_x(w,v)\neq 0$). By setting $w=m^{-1}v$  it is clear that $\sigma_x^m=\mathrm{vol}_{x}^m$ (recall that the normalized Liouville vector $m^{-1}A$ is the unitary normal  to $U^{m}_xM$).  On the other hand, by choosing $w=\partial/\partial v^0$
\[\mathrm{vol}_{x}^m=\sigma_x^m = -m\frac{\partial_{v^0}\lrcorner\mathrm{vol}_{g_x}}{g_{0\mu}(x)v^\mu}=-m \frac{\sqrt{-\det g_x}}{g_{0\mu}(x) v^\mu} \mathrm{d} v^1\wedge \dots\wedge \mathrm{d} v^n=m \frac{\sqrt{-\det g_x}}{|v_0|}  \mathrm{d} v^1\wedge \dots\wedge \mathrm{d} v^n\,,\]
where $v_0:=g_{0\mu}(x)v^\mu$.
\end{proof}

Although it is not strictly necessary, in many cases it is convenient to use the spacetime metric to bring the natural structures of the contangent bundle $T^*M$ to the tangent bundle $TM$. This is so since the geodesic spray and the volume structures derived from the Sasaki metric can be described in terms of the energy, tautological form and the symplectic structure of $T^*M$ (and also because the charged particle treatment is more natural in $T^*M$). We will follow the approach developed by M. Berger in \cite{berger1965lectures}. Taking advantage of the natural isomorphism between $TM$ and $T^*M$ induced by the metric $g$, analogous structures are derived from the cotangent bundle $T^*M$. The tautological one-form of $T^*M$ corresponds to a one-form $\alpha \in \Omega^1(TM)$, with
\[
\alpha(V) = g(v, Tp(V)), \quad V \in T_vTM, \quad p: TM \rightarrow M\,.
\]
Hence, $\mathrm{d}\alpha$ introduces a symplectic form on $TM$ (non-degenerate and closed), essential for the formulation of Hamiltonian dynamics.  The Liouville vector field $L $ and the tautological one-form $\alpha$ satisfy 
\[
L \lrcorner \mathrm{d}\alpha = -\mathrm{d} E,
\]
where the energy function $E: TM \rightarrow \mathbb{R}$ is defined by:
\[
E(v)=\frac{1}{2} g(v,v) \quad \text{or equivalently} \quad \alpha(L )=2E.
\]

\begin{proposition}
Let $\mathrm{vol}_{\textsf{g}}$ denote the volume form on the tangent bundle $TM$ induced by the Sasaki metric $\textsf{g}$, and let $\alpha$ represent the tautological one-form on $TM$ associated with the metric $g$ on the base manifold $M$. Then, the volume form $\mathrm{vol}_{\textsf{g}}$ and the volume form $\mathrm{vol}_m$ on the mass-shell submanifold $U^{m}M$ are given by:
\[
\mathrm{vol}_{\textsf{g}} = c  (\mathrm{d}\alpha)^{1+n},
\quad 
\mathrm{vol}_m= c_m \alpha_m \wedge (\mathrm{d}\alpha_m)^{n} \,,
\]
where $\alpha_m$ is the pullback  of $\alpha$ to $U^{m}M$ and the constants $c$ and $c_m$  are given by
\[
c=\frac{(-1)^{\binom{1+n}{2}}}{(1+n)!}\,,\quad c_m=\frac{(-1)^{\binom{1+n}{2}}}{n!m} \,.
\]
Furthermore, the volume form $\mathrm{vol}_m$ can be disintegrated along the fibers $U^{m}_xM$ for each $x \in M$. Denoting by $\sigma_x^m$ the canonical volume measure on the hyperbolic space $U_x^{m}M$ within the Minkowski space $(T_xM, g_x)$, we obtain the integral formula:
\[
\int_{U^{m}M}  f \, \mathrm{vol}_m =\int_{M} \left(\int_{U^{m}_xM} (f\upharpoonright U^{m}_{x}M)\, \sigma_x^{m}\right) \mathrm{vol}_{g},
\]
where $\mathrm{vol}_{g}$ is the Riemannian volume form on $M$ induced by $g$.
\end{proposition}{\label{teorema_integralesUM}}

\begin{proof} The same steps followed in the Riemannian case discussed in \cite{besse2012manifolds} allow us to show that
$\mathrm{vol}_{\textsf{g}} = c  (\mathrm{d}\alpha)^{1+n}$. To compute $\mathrm{vol}_m$ notice that $m^{-1}A$ is the unit normal to $U^{m}M$ and
\begin{equation}
    m^{-1}A\lrcorner\mathrm{vol}_{\textsf{g}}= m^{-1}c  A \lrcorner (\mathrm{d}\alpha)^{1+n}=
   \frac{c(1+n)}{m}  \alpha\wedge (\mathrm{d}\alpha)^{n} \,, \label{mA}
\end{equation}
where we have used that 
\[
A\lrcorner \mathrm{d}\alpha =\alpha\,.
\]
Therefore, the result follows by pulling-back \eqref{mA}  to $U^{m}M$
\[
\mathrm{vol}_m=\frac{c(1+n)}{m}  \alpha_m\wedge (\mathrm{d}\alpha_m)^{n}\,.
\]

\end{proof}

Notice that, in addition to $\mathrm{vol}_m$, the manifold $U^{m}M$ is also equipped with a canonically defined  $2n$-form 
\begin{align*}
    \omega_{m} := L_{m} \lrcorner \mathrm{vol}_{m} \ .
\end{align*}

The previous discussion leads to the following result.

\begin{corollary}
For a $(1+n)$-dimensional space-time $(M,g)$, the  $(2n+1)$-dimensional future observer bundle $U^{m}M$ is equipped with a volume form  and a $2n$-form,
\[
\mathrm{vol}_m = c_m  \alpha_m \wedge (\mathrm{d}\alpha_m)^n \quad\textrm{ and } \quad  \omega_m=L _m \lrcorner \mathrm{vol}_m= -m^2c_m  (\mathrm{d}\alpha_m)^n,
\]
that satisfy
\[\pounds_{L _m}\mathrm{vol}_m =0,\quad \pounds_{L _m} \omega_m =0,\quad  L _m \lrcorner\omega_m=0,\quad \mathrm{d} \omega_m=0.\]
\end{corollary}

\begin{proof}
The assertion $L_m \lrcorner\omega_m = L_m \lrcorner (L_m\lrcorner\mathrm{vol}_m) = 0$ directly follows from the definition. The remaining properties can be derived from:
\[
\mathrm{d} \omega_m = \mathrm{d}(L_m\lrcorner \mathrm{vol}_m) = \pounds_{L _m}\mathrm{vol}_m = c_{m}
\pounds_{L _m}(\alpha_m \wedge (\mathrm{d}\alpha_m)^{1+n}) = 0,
\]
where the last equality is obtained by making use of $\pounds_{L_m} \alpha_m = 0$, justified by:
\[
L \lrcorner \alpha = 2E \, \textrm{  and  } \, L \lrcorner \mathrm{d}\alpha = -\mathrm{d}E \Rightarrow \pounds_{L } \alpha = \mathrm{d}E.
\]
Given that $L$ is tangent to $U^{m}M$, allowing us to pull back $\pounds_{L} \alpha = \mathrm{d}E$ to $U^{m}M$,  the equation $\pounds_{L_m} \alpha_m = 0$ holds because $E$ is constant on $U^{m}M$. 
\end{proof}

\subsection{Diffusion Equation on the Observer Bundle}

In this section, we examine the diffusion equation within the observer bundle $U^{m}M$, focusing on the integration of the $(2n)$-form $\omega_m$ over hypersurfaces and its implications for volume conservation under the dynamics defined by $L _m$, the Liouville vector field . When $L _m$ is tangent to $\mathcal{S}$, the condition $L _m \lrcorner \omega_m=0$ ensures
\[
\int_{\mathcal{S}} \omega_m = 0 \ .
\]
Conversely, if $L _m$ intersects $\mathcal{S}$ transversely, $\omega_m$ acts as a volume form on $\mathcal{S}$. The closure of $\omega_m$, denoted by $\mathrm{d}\omega_m=0$, guarantees a consistent volume assignment across any two diffeomorphically related hypersurfaces, $\mathcal{S}_1$ and $\mathcal{S}_2$, via the flow induced by $L _m$. To establish this, consider a tube $\mathcal{T}$ formed by Lie dragging a compact hypersurface $\mathcal{S}_1$ along $L _m$ to a second boundary $\mathcal{S}_2$. The invariance of volume under such transformations follows from:
\begin{align}
    0 = \int_{\mathcal{T}} \mathrm{d} \omega_m = \int_{\mathcal{S}_2} \omega_m - \int_{\mathcal{S}_1} \omega_m,
\end{align}
indicating that the volume enclosed by any hypersurface $\mathcal{S}$ shaped through this process remains invariant.

The hypersurfaces $\mathcal{S}$ (transverses to $L_m$) may be endowed with a  volume form 
\[m^{-1} f \omega_m \ ,\]for any smooth function $f>0$: the so called \textit{one-particle distribution function} on $U^mM$. The normalization factor $m^{-1}$ comes from the fact that \[\textsf{g}(L,L)\upharpoonright U^mM=-m^2 \ . \] Physically, the quantity  
\[N(\mathcal{S}):=m^{-1}\int_{\mathcal{S}} f \omega_m \ , \]
provides the average of  particle trajectories that pass through $\mathcal{S}$ (that is, trajectories with tangent vectors belonging to  $\mathcal{S}$). 

We observe that
\begin{align}\label{tube}
    \mathrm{d}(f \omega_m) = \pounds_{L_m} (f \, \text{vol}_{m}) = (\pounds_{L_m}f)  \text{vol}_{m},
\end{align}
indicating that, irrespective of $f$, the following integral relation holds:
\begin{align}
    \int_{\mathcal{T}} \pounds_{L_m} f \, \text{vol}_{m} = \int_{\mathcal{S}_2} f \omega_m  - \int_{\mathcal{S}_1} f \omega_m = m \Big( N(\mathcal{S}_2)-  N(\mathcal{S}_1)\Big).   
\end{align}
Enforcing $\pounds_{L_m} f = 0$ guaranteed that the volume form $f \omega_m$ consistently assigns identical volumes to all hypersurfaces $\mathcal{S}$ within the tube $\mathcal{T}$:
\begin{equation}
    \int_{\mathcal{S}_2} f \omega_m  = \int_{\mathcal{S}_1} f \omega_m.\label{conservation}
\end{equation}
Therefore, there is no net change in the average number of particle trajectories passing through the surfaces $\mathcal{S}_1$ and $\mathcal{S}_2$. This condition, encapsulated by the equation 
\[\pounds_{L_m} f = 0,\] is recognized as the Vlasov (or Liouville) equation. However, adherence to the conservation law \eqref{conservation}  does not strictly necessitate the Vlasov equation; it suffices for $f$ to satisfy:
\begin{equation}
     \int_{\mathcal{T}} \pounds_{L_m} f \, \text{vol}_{m} = 0,
\end{equation}
leveraging theorem \ref{teorema_integralesUM} to express this as:
\[
 \int_{\mathcal{T}} (\pounds_{L_m} f) \, \text{vol}_{m} = \int_{x \in M} \left( \int_{U^{m}_xM} (\pounds_{L_m} f\upharpoonright{U^{m}_xM}) \, \sigma^m_x \right) \text{vol}_{g}=0,
\]
where requiring that  the integral over $U^{m}_xM$ vanishes  for every $x \in M$ guarantees the conservation law \eqref{conservation}, i.e.
\[
\int_{U^{m}_xM} (\pounds_{L_m} f\upharpoonright{U^{m}_xM}) \, \sigma^m_x =0  \Rightarrow \int_{\mathcal{S}_2} f \omega_m  = \int_{\mathcal{S}_1} f \omega_m\,.
\]
This might be obtained for instance with, 
\begin{equation}
    \pounds_{L_m} f = \frac{\sigma^2}{2} \Delta^{\mathrm{ver}}_m f, \label{FPecuacion}
\end{equation}
where the operator $\Delta^{\mathrm{ver}}_m$ acts on $f\upharpoonright U^{m}_xM$ similarly to the Laplace-Beltrami operator on $U^{m}_xM$, ensuring the integral's nullity via integration by parts. This formulation aligns with the Fokker--Planck equation, underscoring geometric diffusion processes as developed by Franchi and Jan \cite{franchi_jan}, and introduced by Calogero \cite{calogero2011kinetic}.

\subsection{Current densities and conservation laws}

In practice, the hypersurfaces $\mathcal{S}$ appearing in equation \eqref{conservation} that we will use in the following have the form 
\begin{align}
    \mathcal{S}_\Sigma =\pi^{-1}_m(\Sigma) = \{ v\in U^{m}M \ : \ \pi_m(v)\in \Sigma \} \ ,
\end{align}
and are constructed from spacelike hypersurfaces $\Sigma\subset M$. Notice that, if $f$ decays rapidly enough, equation \eqref{conservation} holds even when $\Sigma$ is a (non-compact) Cauchy surface. We can split  the integrals in \eqref{tube} over $ \mathcal{S}_\Sigma$ in a similar way as the one in theorem \ref{teorema_integralesUM} to get 
\begin{align}
\label{mean_particle}
    \int_{ \mathcal{S}_\Sigma} f\omega_m = \int_{x\in \Sigma} \left( \int_{v\in U^{m}_xM} g(\mathsf{n}_x,v) f(v) \ \sigma^m_x(v)  \right) \text{vol}_\Sigma(x) \ ,
\end{align}
where $\mathsf{n}$ is the future-directed normal to $\Sigma$. Motivated by the previous expression, associated  with $f$, we define the \textit{particle current density} $\mathsf{J}$, i.e. the vector field $\mathsf{J}\in \mathfrak{X}(M)$ [$M\ni x\mapsto \mathsf{J}_x\in T_xM$]
\begin{align}
   \mathsf{J}(\alpha) = m^{-1}\int_{U^{m}_x} F^\alpha  f  \sigma^m_x = m^{-1} \int_{v\in U^{m}_x} \alpha(v) \ f(v) \ \sigma^m_x(v) \ , \quad \forall \alpha \in T_x^*M\,,
\end{align}
where $v\mapsto F^\alpha(v)=\alpha(v)$ is a smooth function on $U^m_xM$ for any $\alpha\in T_x^*M$. 
In this way, we can consider the map 
\[
\Sigma\ni x\mapsto g_x(\mathsf{J},\mathsf{n})=\mathsf{J}(\mathsf{n}^\flat)(x) =\int_{\in U^{m}_xM} F^{\mathsf{n}^\flat_x}  f  \sigma^m_x =  \int_{v\in U^{m}_x} g(\mathsf{n}_x,v) \ f(v) \ \sigma^m_x(v) \ ,
\]
where $\mathsf{n}^\flat$ is the one form $\mathsf{n}^\flat(\cdot)=g(\mathsf{n},\cdot)$. Therefore, if $\Sigma_1$, $\Sigma_2$ are Cauchy hypersurfaces, then denote $(\mathcal{S}_{\Sigma_1}, \mathcal{S}_{\Sigma_2})_m$ the tube $\mathcal{T}$ generated by $\mathcal{S}_{\Sigma_1}$ to $ \mathcal{S}_{\Sigma_2}$ by the flow of $L_m$. Then the following holds
\begin{align}
\label{J_conserved}
 0=\int_{(\mathcal{S}_{\Sigma_1}, \mathcal{S}_{\Sigma_2})_m} \pounds_{L_m} f \ \text{vol}_{m} = \int_{\Sigma_2} g(\mathsf{J}, \mathsf{n}_2) \ \text{vol}_{\Sigma_2} -\int_{\Sigma_1} g(\mathsf{J}, \mathsf{n}_1) \  \text{vol}_{\Sigma_1} \ ,
\end{align}
where, $\mathsf{n}_1$, $\mathsf{n}_2$ are the future-directed normals of $\Sigma_1$, $\Sigma_2$. In the case $\Sigma_1$, $\Sigma_2$ are non-compact, $f$ should decay rapidly enough.

In terms of local coordinates 
\[
\mathsf{J}(x)=m^{-1}\left( \int_{v\in U^{m}_x} v^\mu \ f(v) \ \sigma^m_x(v)\right) \frac{\partial}{\partial x^\mu}(x) \,.
\]

\begin{theorem} If $f$ satisfies the FP equation $\left(L_m -\frac{\sigma^2}{2}\Delta^{\mathrm{ver}}_m\right) f =0$ then 
$\mathrm{div} \,\mathsf{J} =0$.
\end{theorem}
\begin{proof}
A direct argument (see, for example \cite{ehlers}), in which it suffices to express Stokes' theorem in terms of the semi-Riemannian metrics involved, informs us that if 
\[
\mathsf{J}(x)=\left( \int_{v\in U^{m}_x} v^\mu \ f(v) \ \sigma^m_x(v)\right) \frac{\partial}{\partial x^\mu}(x) \ ,
\]
the following identity holds  
\[
\mathrm{div} \mathsf{J} = m^{-1}\int_{U^m_xM} (\pounds_{L_m} f ) \, \sigma^m_x\,.
\] 
Hence, if $f$ satisfies the FP equation 
    \begin{align*}
    (\text{div} \, \mathsf{J})(x) =& m^{-1} \int_{U^m_xM} (\pounds_{L_m} f ) \, \sigma^m_x = \frac{\sigma^2}{2m} \int_{\mathbb{H}_x} (\Delta^{\mathbb{H}_x} f) \, \sigma^m_x =0\,,
\end{align*}
where in the last equality we have used integration by parts. 
\end{proof}
The property $\mathrm{div}\, \mathsf{J}=0$ is an expression of the conservation of the average number of particle world lines crossing a Cauchy hypersurface $\Sigma$:
\[
\mathsf{N}(\Sigma):= N(\mathcal{S}_\Sigma)=m^{-1}\int_{\mathcal{S}_\Sigma} f\omega_m= m^{-1}\int_{\Sigma} g(\mathsf{J}, \mathsf{n}_\Sigma) \  \text{vol}_{\Sigma}\,,
\]
where $\mathsf{n}_\Sigma$ is the (future pointing) unit normal to $\Sigma$. 

We can also define an \textit{entropy current} associated with a solution $f$ of the FP equation  as the vector field [$M\ni x\mapsto \mathsf{S}_x\in T_xM$]
\begin{align*}
\mathsf{S}(\alpha) = -\frac{k_B}{m}\int_{U^m_xM} F^\alpha  f \log{f} \  \sigma^m_x = -\frac{k_B}{m}\int_{v\in U^m_xM}\alpha(v)  f(v) \log{f(v)} \  \sigma^m_x(v) \ ,\quad \alpha\in T^*_xM\,.
\end{align*}
where $k_B$ is the Boltzmann constant. In local coordinates 
\begin{align*}
\mathsf{S}=S^\mu\frac{\partial}{\partial x^\mu}\,,\quad S^\mu(x) =-\frac{k_B}{m}\int_{v\in U^m_xM} v^\mu f(v) \log{f(v)} \  \sigma^m_x(v) \,.
\end{align*}
\begin{theorem} If $f>0$ satisfies the FP equation $\left(L_m -\frac{\sigma^2}{2}\Delta^{\mathrm{ver}}_m\right) f =0$ then 
$\mathrm{div} \,\mathsf{S} \geq  0$.
\end{theorem}
\begin{proof}
    \begin{align*}
    (\text{div} \,\mathsf{S})(x) =&  - \frac{k_B}{m}\int_{U^m_xM} \pounds_{L_m} ( f \log{f} ) \, \sigma^m_x  = - \frac{k_B}{m}\int_{U^m_xM} (\pounds_{L_m} f )\left( \log{f} +1 \right) \, \sigma^m_x  
     \\=& -\frac{k_B\sigma^2}{2m} \int_{\mathbb{H}_x} (\Delta^{\mathbb{H}_x} f) \left( \log{f} +1 \right) \, \sigma^m_x 
    =-\frac{k_B\sigma^2}{2m} \int_{\mathbb{H}_x} (\Delta^{\mathbb{H}_x} f) \log{f} \, \sigma^m_x \\
    =& \frac{k_B\sigma^2}{2m}\int_{\mathbb{H}_x} \frac{|| \mathrm{grad}\, f ||^2}{f} \sigma^m_x  \geq 0 \ ,
\end{align*}
since the metric on $\mathbb{H}_x$ is Riemannian and $f>0$.
\end{proof}
The entropy current associates a total entropy $S_\Sigma$ to any (oriented) hypersuface $\Sigma\subset M$ through the expression
\[
\mathsf{s}(\Sigma) := \int_\Sigma g(\mathsf{S},\mathsf{n})\, \mathrm{vol}_\Sigma\,,
\]
where, as before, $\mathsf{n}_\Sigma$ denotes the normal to $\Sigma$.

Similarly, one can also introduce the \textit{energy-momentum tensor}
\begin{align*}
    \mathsf{T}(\alpha, \beta)=T(\beta,\alpha) = m^{-1}\int_{U^{m}_xM} F^\alpha F^\beta  f \sigma^m_x
    =
    m^{-1}\int_{v\in U^{m}_xM} \alpha(v) \ \beta(v)\ f(v) \ \sigma^m_x(v) \ , 
\end{align*}
for all $\alpha$, $\beta \in T_x^*M\,.$ In local coordinates,
\begin{align*}
\mathsf{T}=T^{\mu\nu}\frac{\partial}{\partial x^\mu} \otimes\frac{\partial}{\partial x^\nu}\,,\quad T^{\mu\nu}(x) =m^{-1}\int_{v\in U^m_xM} v^\mu v^\nu f(v) \ \sigma^m_x(v) \,.
\end{align*}

\begin{theorem} If $f$ satisfies the FP equation $\left(L_m -\frac{\sigma^2}{2}\Delta^{\mathrm{ver}}_m\right) f =0$ then 
$\mathrm{div} \,\mathsf{T} = \frac{\sigma^2}{2} n \mathsf{J} $
\end{theorem}
\begin{proof} Given any $\beta\in T^*_xM$, we have that 
    \begin{align*}
        \left( \mathrm{div} \,\mathsf{T} \right)(\beta) &= m^{-1}\int_{ U^m_xM}  (\pounds_{L_m} f) F^\beta \, \sigma^m_x=  \frac{\sigma^2}{2m} \int_{\mathbb{H}_x} (\Delta^{\mathbb{H}_x} f) F^\beta \, \sigma^m_x  \\
        &
        =\frac{\sigma^2}{2m}\int_{\mathbb{H}_x} f \,(\Delta^{\mathbb{H}_x} F^\beta) \, \sigma^m_x 
        =  \frac{n\sigma^2}{2m} \int_{\mathbb{H}_x} f \,  F^\beta \, \sigma^m_x 
      \\&= \frac{n\sigma^2}{2}  \mathsf{J}(\beta)
    \end{align*}
where we have used that $F^\beta: U^m_xM\rightarrow \mathbb{R}$, $v\mapsto F(v)= \beta(v)$, satisfies $\Delta^{\mathbb{H}_x} F^\beta =n F^\beta$. 
\end{proof}

\begin{corollary}
    If $f$ satisfies the FP equation $\left(L_m -\frac{\sigma^2}{2}\Delta^{\mathrm{ver}}_m\right) f =0$ then 
$\mathrm{div}(\mathrm{div} \,\mathsf{T}) = 0 $. 
\end{corollary}

\section{Some relevant spacetimes}{\label{sec:examples}}

In this section, we use the results presented in the previous sections to obtain the Fokker--Planck equations for one-particle distribution functions preserving the underlying symmetries in two physically relevant spacetimes.

\subsection{Flat Friedmann–Lemaître–Robertson–Walker cosmologies}

The line-element  of the flat Friedmann–Lemaître–Robertson–Walker (FLRW)  space time $(\mathbb{R}^4,g)$  is
\begin{align*}
    ds^2 = -(\mathrm{d}x^0)^2 + a^2 \delta_{ij} \mathrm{d}x^i \mathrm{d}x^j \ ,
\end{align*}
where $x^\mu$ are Cartesian coordinates of $\mathbb{R}^4$ and  $a=a(x^0)>0$ is the scale factor. Using global bundle coordinates, the geodesic spray $L_m$ and the vertical Laplacian $\Delta^{\mathrm{ver}}_m$ on $U^mM=\mathbb{R}^4\times \mathbb{R}^3$ are given by
\begin{align*}
    L_m &=  \mathrm{v}^0 \left( \partial_{x^0}  - 2a'a^{-1} v^i \partial_{v^i} \right) + v^i \partial_{x^i} \ , \\
    \Delta^{\mathrm{ver}}_m &= \left( \frac{\delta^{ij}}{a^2} + \frac{v^iv^j}{m^2} \right) \partial_{v^i}\partial_{v^j}+   \frac{3v^i}{m^2} \partial_{v^i}\ ,
\end{align*}
where we have defined \[\mathrm{v}^0=\mathrm{v}^0(x^0,v^i):=\sqrt{m^2+a^2\mathrm{v}^2}\,\quad
\mathrm{v} := \mathrm{v}(v^i) = \sqrt{\delta_{ij}v^i v^j}\, .\]
The FP equation \eqref{FPecuacion} can be written in the form 
\begin{align}
\label{fp_flrw}
    \frac{\partial f}{\partial {x^0}} & + \frac{v^i}{\mathrm{v}^0} \frac{\partial f}{\partial {x^i}}     =  \frac{\sigma^2}{2\mathrm{v}^0} \left( \left(  \frac{\delta^{ij}}{a^2}  + \frac{v^iv^j}{m^2}\right) \frac{\partial^2 f}{\partial {v^i} \partial {v^j}}
  +\frac{3v^i}{m^2} \frac{\partial f}{\partial {v^i}}   \right) 
  + 2\frac{a'}{a}  v^i   \frac{\partial f}{\partial {v^i}} \ .
\end{align}

The isometry group of $(M,g)$ is just the Euclidean group $E(3)$ ($\dim E(3)=6$). The  Killing fields of the metric $g$ are the ones associated with homogeneity and isotropy (translations and rotations):
\begin{align*}
    \xi_i = \frac{\partial}{\partial x^i} \ ,\quad   \xi_{ij} = x^i \frac{\partial}{\partial x^j} - x^j \frac{\partial}{\partial x^i}\ ,
\end{align*}
whose lifts to $U^mM$ are
\begin{align*}
    \xi^{\mathrm{ct}}_i = \frac{\partial}{\partial x^i} \ ,\quad    \xi^{\mathrm{ct}}_{ij} = x^i \frac{\partial}{\partial x^j} - x^j \frac{\partial}{\partial x^i} + v^i \frac{\partial}{\partial v^j} - v^j \frac{\partial}{\partial v^i} \ .
\end{align*}
Forcing $f$ to simultaneously satisfy $ \pounds_{\xi^{\mathrm{ct}}_i} f=0$ and $\pounds_{\xi^{\mathrm{ct}}_{ij}} f=0$ we have 
\begin{align*}
\frac{\partial f}{\partial x^i} = 0\,,\quad   v^i \frac{\partial f}{\partial v^j} - v^j \frac{\partial f}{\partial v^i} = 0\,.
\end{align*}
Hence $f=f(x^0,x^i,v^i)$
must be independent of $x^i$ and the dependence on $v^i$ is only through $\mathrm{v}$ because  $ \pounds_{\xi^{\mathrm{ct}}_{ij}} \mathrm{v}=0$. Therefore, if we are interested in symmetric solutions to the FP equation of we can consider $f(x^0,v^i)=F(x^0, \mathrm{v})$ and rewrite the equation \eqref{fp_flrw} in the simpler form
\begin{align*}
    \frac{\partial F}{\partial x^0 }  - 2 a' a^{-1} \mathrm{v}  \frac{\partial F}{\partial \mathrm{v} }   =  \frac{\sigma^2 }{2  \mathrm{v}^0} 
    \left( \left( \frac{1}{a^2}+ \frac{\mathrm{v}^2}{m^2} \right)  \frac{\partial^2 F}{\partial \mathrm{v}^2 }
    +  \left(\frac{2}{a^2\mathrm{v}}+\frac{3\mathrm{v}}{m^2} \right) \frac{\partial F}{\partial \mathrm{v} } \right)  \ .
\end{align*}
By using the notation introduced in equation  \eqref{conservado}, the previous equation can be further simplified if the norm $p$ of the linear momentum,   \[p=\delta_{ij} C^{\xi_i}C^{\xi_j}=a^2\mathrm{v} \ ,\]
is used instead of $\mathrm{v}$. This is so because
\[
\pounds_{\xi^{\mathrm{ct}}_{ij}} p=0\quad \textrm{ and }\quad \pounds_{L_m} p=0\,.
\]
Using $(x^0,p)$ as reduced coordinates, the FP equation becomes 
\begin{align*}
  \mathrm{v}^0  \frac{\partial F}{\partial x^0 }   =  \frac{\sigma^2 }{2  } 
    \left( \left( a^2+ \frac{p^2}{m^2} \right)  \frac{\partial^2 F}{\partial p^2 }
    + \left(\frac{2a^2}{p}+\frac{3p}{m^2}\right)  \frac{\partial \,\mathrm{F}}{\partial p } \right)  \ ,
\end{align*}
where $\mathrm{v}_0=a^{-1}\sqrt{m^2a^2+p^2}$.

Taking into account that 
\[
\sigma^m_{x}= m \frac{a^3(x^0)}{\mathrm{v}^0(x^0,\mathrm{v})} \, \mathrm{d}v^1\wedge \mathrm{d}v^2 \wedge \mathrm{d}v^3\,,
\]
the currents associated with the solutions of the FP equation are
\begin{align*}
    \mathsf{J}^0(x^0,\bm{x}) &= a^3(x^0) \int_{\mathbb{R}^3} \ f(x^0, \bm{x},\bm{v}) \ \mathrm{d}^3\bm{v} , \\
    \mathsf{J}^i(x^0,\bm{x}) &= a^3(x^0) \int_{\mathbb{R}^3} \  \ \frac{v^i \ f(x^0, \bm{x},\bm{v})}{\mathrm{v}^0(x^0,\bm{v})} \,\mathrm{d}^3\bm{v} , \\
    \mathsf{S}^0(x^0, \mathbf{x}) &= -\frac{k_Ba^3(x^0)}{m} \int_{\mathbb{R}^3}  \ f(x^0, \bm{x},\bm{v}) \log{f(x^0, \bm{x},\bm{v}) } \ \mathrm{d}^3\bm{v}  ,\\
    \mathsf{S}^i(x^0, \mathbf{x}) &= -\frac{k_Ba^3(x^0)}{m}  \int_{\mathbb{R}^3}  \ \frac{ v^i \ f(x^0, \bm{x},\bm{v}) \log{f(x^0, \bm{x},\bm{v}) }  }{\mathrm{v}^0(x^0,\bm{v})} \ \mathrm{d}^3\bm{v}  ,
\end{align*}
and the conserved average number of occupied trajectories $\mathsf{N}(\Sigma_{t})$ of the Cauchy hypersurfaces $\Sigma_t = \{x^0=t\}$ is given by
\[
    \mathsf{N}(\Sigma_{t}) = a^3(t) \int_{\mathbb{R}^3} \mathsf{J}^0(t, \bm{x}) \, \mathrm{d}^3\bm{x} = a^6(t)\int_{\mathbb{R}^3\times \mathbb{R}^3} f(t, \bm{x},\bm{v}) \ \mathrm{d}^3\bm{x} \, \mathrm{d}^3\bm{v} \ .
    \]
For symmetric solutions $f(t, \bm{x},\bm{v})=F(t,\mathrm{v})=F(t,p)$  (with the usual abuse of notation)
\[
\mathsf{J}^0(x^0)=\frac{4\pi}{a^3(x^0)} \int_0^\infty F(x^0, p)\, p^2\mathrm{d}p \quad \textrm{ and } \quad \mathsf{J}^i(x^0)=0\,.
\]
In this case
\[
0=\mathrm{div}\, \mathsf{J} =\frac{1}{a^3(x^0)}\frac{\partial}{\partial x^0}\big(a^3(x^0)\mathsf{J}^0(x^0)\big) \ .
\]
Hence, as pointed out in \cite{calogero2011kinetic}, 
\[
a^3(t_1)\mathsf{J}^0(t_1)= a^3(t_2)\mathsf{J}^0(t_2)\,,\quad \forall t_1,t_2\,.
\]
Notice that, for symmetric solutions, the integral in the definition of 
\[\mathsf{N}(\Sigma_{t})=
\int_{\Sigma_t} g(\mathsf{J}, \mathsf{n}_{\Sigma_t}) \  \text{vol}_{\Sigma_t}=\int_{\mathbb{R}^3} a^3(t)\,\mathsf{J}^0(t) \,\mathrm{d}^3\bm{x} \ ,
\]
diverges (because $\Sigma_t$ is not compact) but one can renormalize the relevant objects. For example, as we have pointed out,   
\begin{align*}
    \mathsf{n}(\Sigma_t) = a^3(t) \mathsf{J}^0(t)= 4\pi a^6(t) \int_0^\infty  \mathrm{v}^2 F(t,\mathrm{v}) \, \mathrm{d}\mathrm{v} =4\pi \int_0^\infty  p^2 F(t,p) \, \mathrm{d}p \ ,
\end{align*}
does not depend on $t$.

\subsection{Exterior Schwarzschild spacetime}
The exterior Schwarzschild spacetime $(M,g)$ describes the gravitational field outside a spherical body of mass $r_s/2$ where $r_s$ is the Schwarzschild  radius. Topologically 
\[M=\mathbb{R}\times (r_s,\infty)\times \mathbb{S}^2\,, \] 
and, in spherical coordinates, the line-element is given by 
\begin{align*}
    ds^2 = -\left( 1 - \frac{r_s}{r}\right) dt^2 + \left( 1 - \frac{r_s}{r}\right)^{-1} dr^2 + r^2 (d\theta^2 + \sin^2{\theta} d\varphi^2) \ .
\end{align*}

The geodesic flow vector field on $U^mM$ is given by
\begin{align*}
    L_m &=  \sqrt{ \frac{r}{r-r_s} \left( m^2 + \mathrm{v}^2 \right)} \ \partial_t + v^r \partial_{r} + v^\theta \partial_{\theta}+ v^\varphi \partial_{\varphi}\\
    &- \frac{1}{2r^2}\Big( r_sm^2 + (3r_s-2r)r^2 \left( (v^\theta)^2  + \sin^2{\theta} (v^\varphi)^2 \right)  \Big)\partial_{v^r}  \\
    &- \frac{1}{r}\left( 2 v^r v^\theta - r\sin{\theta} \cos{\theta} (v^\varphi)^2 \right) \partial_{v^\theta} - \frac{1}{r}\left( 2 v^r v^\varphi  + 2r v^\theta v^\varphi \cot{\theta} \right)  \partial_{v^\varphi} \ ,
\end{align*}
where we have defined 
\[\mathrm{v}^2 := \frac{r}{r-r_s} (v^r)^2 + r^2 \Big((v^\theta)^2 + \sin^2{\theta} (v^\varphi)^2\Big)\,.\]
The isometry group  of $(M,g)$ is $\mathbb{R}\times O(3)\times \mathbb{Z}_2$ (time translations, the orthogonal group in three dimensions, and time reversal). The (globally defined, smooth) Killing vector fields of this metric is  given, in local coordinates, by 
\begin{align*}
    \xi_0 &= \frac{\partial}{\partial t} \ , \\
    \xi_1 &= \sin{\varphi}\frac{\partial}{\partial \theta} + \cot{\theta} \cos{\varphi}\frac{\partial}{\partial \varphi}  \ , \\
    \xi_2 &= \cos{\varphi}\frac{\partial}{\partial \theta} - \cot{\theta} \sin{\varphi}\frac{\partial}{\partial \varphi} 
         \ , \\
    \xi_3 &= \frac{\partial}{\partial \varphi} \ .
\end{align*}
The vector field  $\xi_0$ is timelike and the  fields $\xi_i$, $i=1, 2, 3$, are tangent to the spheres $\{t=t_0,r=r_0\}$ and generate a $\mathfrak{so}(3)$ algebra. Notice that 
\[
g(\xi_1,\xi_1)=r^2(\sin^2\varphi+\cos^2\theta\cos^2\varphi)\,,\, g(\xi_2,\xi_2)=r^2(\cos^2\varphi+\cos^2\theta\sin^2\varphi)\,,\,  g(\xi_3,\xi_3)=\sin^2\theta.
\]
Then $\xi_1$, $\xi_2$, and $\xi_3$ vanishes at the intersection of the sphere with the axis of the rotation that each of them generate. Their lifts to $U^mM$ are given by
\begin{align*}
    \xi^{\mathrm{ct}}_0 &= \frac{\partial}{\partial t} \ , \\
    \xi^{\mathrm{ct}}_1 &= \sin{\varphi}\frac{\partial}{\partial \theta} + \cot{\theta} \cos{\varphi}\frac{\partial}{\partial \varphi} 
    +v^\varphi \cos \varphi \frac{\partial}{\partial v^\theta} -\left(  v^\theta (1+\cot^2 \theta)\cos\varphi+v^\varphi\cot \theta\sin\varphi \right)\frac{\partial}{\partial v^\varphi}  \ , \\
    \xi^{\mathrm{ct}}_2 &= \cos{\varphi}\frac{\partial}{\partial \theta} - \cot{\theta} \sin{\varphi}\frac{\partial}{\partial \varphi} 
     -v^\varphi \sin \varphi \frac{\partial}{\partial v^\theta} +\left(  v^\theta (1+\cot^2 \theta)\sin\varphi-v^\varphi\cot \theta\cos\varphi \right)\frac{\partial}{\partial v^\varphi} 
    \ , \\
    \xi^{\mathrm{ct}}_3 &= \frac{\partial}{\partial \varphi} \ .
\end{align*}
Hence, imposing \[\pounds_{\xi^{\mathrm{ct}}_k}f=0,\quad k=0,1,2,3,\] 
a straightforward computation implies that the function $f=f(t,r,\theta,\varphi,v^r,v^\theta,v^\varphi)$ must have the form
\begin{align*}
    f = F(r, v^r, \ell) \ ,
\end{align*}
where $\ell$ denotes the \textit{angular momentum} that appears in the change of chart  $(r,v^r,\theta,v^\theta,v^\varphi)\leftrightarrow (r,v^r,\theta,\ell,\psi)$ defined though
\begin{align*}
 \ell^2 &= r^4 \big((v^\theta)^2 + \sin^2{\theta} (v^\varphi)^2\big)\,,\quad \psi =\arctan\left(v^\varphi \sin\theta /v^\theta\right).
\end{align*}
Using the notation introduced in \eqref{conservado}:  
\[\ell^2=\delta_{ij}C^{\xi_{i}}C^{\xi_{j}} \ , \]
and hence $\ell$ satisfies  
\[\pounds_{\xi^{\mathrm{ct}}_k}\ell=0\,,\quad \pounds_{L_m}\ell=0\,.\]
Notice that, contrary to what happens when imposing $\pounds_{\xi_i}h=0$ on a function $h:\mathbb{R}^3\rightarrow \mathbb{R}$, the conditions $\pounds_{\xi^{\mathrm{ct}}_i}f=0$ eliminate the dependence of $f$ on three of its arguments.

Using these coordinates, the FP equation reduces to
\begin{align}  \label{SchReducido}
 & \left(v^r  \partial_r +\frac{(2r-3r_s)\ell^2-m^2r_sr^2}{2r^4}\partial_{v^r}\right) F
     =\\
 &= \frac{\sigma^2}{2} 
     \Bigg( \left(r^2 +\frac{\ell^2}{m^2}\right)  \partial^2_\ell +\left(\frac{r^2}{\ell}+\frac{3\ell}{m^2}\right)\partial_\ell
     +\left(1-\frac{r_s}{r}+\frac{(v^r)^2}{m^2}\right)\partial^2_{v^r}+ \frac{2\ell v^r}{m^2}\partial_\ell\partial_{v^r} +\frac{3v^r}{m^2}\partial_{v^r}  \Bigg) F   \ ,   \nonumber
\end{align}
in full agreement (when $m=1$) with corollary 4.2 of reference \cite{franchi_jan}.

\subsection{Nariai spacetime}

As discussed in reference \cite{schleich2010simple}, if $\dim M=4$, the Birkhoff's theorem states that the only locally spherically symmetric solutions to $\mathrm{Ric}=\Lambda g$ are \textit{locally} isometric either to one of the Schwarzschild-de Sitter (anti-de Sitter) family of solutions 
\begin{align}
    ds^2_{\mathrm{SdS}} = - \left(1-\frac{r_s}{r}-\frac{\Lambda r^2}{3}\right) dt^2 + \left(1-\frac{r_s}{r}-\frac{\Lambda r^2}{3}\right)^{-1} dr^2 + r^2 (d\theta^2 + \sin^2{\theta} d\varphi^2) \ ,
    \label{schdsit}
\end{align}
or to the Nariai spacetime 
\begin{align*}
   M =\mathbb{R}\times \mathbb{S}^1\times \mathbb{S}^2\,,\quad   ds^2 = - dt^2 + \frac{1}{\Lambda}\cosh^2(\sqrt{\Lambda} t) dx^2 + \frac{1}{\Lambda} (d\theta^2 + \sin^2{\theta} d\varphi^2) \ ,
\end{align*}
for which $\Lambda$ is stricly positive. The Nariai spacetime is the semi-Riemannian product of the 1+1 de Sitter space $\mathrm{dS}_2(\Lambda)$ and a round sphere of curvature $\Lambda$. Hence, its isometry group is $O(2,1)\times O(3)$. The (globally defined, smooth) Killing vector fields of the Nariai metric are given, in local coordinates, by:  
\begin{align*}
    \xi_1 &= \sin{\varphi}\frac{\partial}{\partial \theta} + \cot{\theta} \cos{\varphi}\frac{\partial}{\partial \varphi} 
      \ , \\
    \xi_2 &= \cos{\varphi}\frac{\partial}{\partial \theta} - \cot{\theta} \sin{\varphi}\frac{\partial}{\partial \varphi} 
    \ , \\
    \xi_3 &= \frac{\partial}{\partial \varphi} \ ,
    \\
    \kappa_1 &= \sin(x)\frac{\partial}{\partial t} +\sqrt{\Lambda}\cos(x)\tanh(\sqrt{\Lambda} t) \frac{\partial}{\partial x}  \ ,\\
        \kappa_2 &= -\cos(x)\frac{\partial}{\partial t} +\sqrt{\Lambda}\sin(x)\tanh(\sqrt{\Lambda} t) \frac{\partial}{\partial x} \ ,\\  
                 \kappa_3 &= \frac{\partial}{\partial x} \ .          
\end{align*}
The Lie algebra of the isometry group of   the Nariai  spacetime is isomorphic to $\mathfrak{sl}(2,\mathbb{R})\times \mathfrak{so}(3)$. The vector fields $\xi_1$, $\xi_2$, and $\xi_3$ are responsible for the $\mathfrak{so}(3)$ sector and $\kappa_1$, $\kappa_2$, and $\kappa_3$ correspond to $\mathfrak{sl}(2,\mathbb{R})$:
\[
[\kappa_1,\kappa_2]=\Lambda \kappa_3\,,\quad [\kappa_2,\kappa_3]=-\kappa_1\,,\quad [\kappa_3,\kappa_1]= -\kappa_2\,,
\]
as can be easily seen, for example, using the basis 
\[b_1=\Lambda^{-1/2}\kappa_2-\kappa_3\,, b_2=\Lambda^{-1/2} \kappa_1\,, b_3=\Lambda^{-1/2}\kappa_2+\kappa_3\,.\]
In addition to the angular momentum of $\mathfrak{so}(3)$, we have the following conserved quantities derived from $\mathfrak{sl}(2,\mathbb{R})$:
\begin{align*}
 C^{\kappa_1}&= -\sin(x) v^t+\frac{\cos(x)\sinh(\sqrt{\Lambda}t)\cosh(\sqrt{\Lambda}t)}{\sqrt{\Lambda}}  v^x    \ ,\\
 C^{\kappa_2}&= \cos(x) v^t+\frac{\sin(x)\sinh(\sqrt{\Lambda}t)\cosh(\sqrt{\Lambda}t)}{\sqrt{\Lambda}}  v^x    \ ,\\ 
C^{\kappa_3}&= \frac{\cosh^2(\sqrt{\Lambda}t)}{\Lambda} v^x \ .            
\end{align*}
Notice that
\[
\big( C^{\kappa_1}-\sqrt{\Lambda}\cos(x)\tanh(\sqrt{\Lambda}t) C^{\kappa_3} \big)^2+\big( C^{\kappa_2}-\sqrt{\Lambda}\sin(x)\tanh(\sqrt{\Lambda}t) C^{\kappa_3} \big)^2=(v^{t})^2\,.
\]
The lifts to $U^mM$ of the Killing fields are given by
\begin{align*}
    \xi^{\mathrm{ct}}_1 &= \sin{\varphi}\frac{\partial}{\partial \theta} + \cot{\theta} \cos{\varphi}\frac{\partial}{\partial \varphi} 
    +v^\varphi \cos \varphi \frac{\partial}{\partial v^\theta} -\left(  v^\theta (1+\cot^2 \theta)\cos\varphi+v^\varphi\cot \theta\sin\varphi \right)\frac{\partial}{\partial v^\varphi}  \ , \\
    \xi^{\mathrm{ct}}_2 &= \cos{\varphi}\frac{\partial}{\partial \theta} - \cot{\theta} \sin{\varphi}\frac{\partial}{\partial \varphi} 
     -v^\varphi \sin \varphi \frac{\partial}{\partial v^\theta} +\left(  v^\theta (1+\cot^2 \theta)\sin\varphi-v^\varphi\cot \theta\cos\varphi \right)\frac{\partial}{\partial v^\varphi} 
    \ , \\
    \xi^{\mathrm{ct}}_3 &= \frac{\partial}{\partial \varphi} \ ,\\
              \kappa^{\mathrm{ct}}_1 &= 
              \sin(x)\frac{\partial}{\partial t} +\sqrt{\Lambda}\cos(x)\tanh(\sqrt{\Lambda} t) \frac{\partial}{\partial x} 
               \\&+v^x\cos(x)\frac{\partial}{\partial v^t}
                +\sqrt{\Lambda}\left(-v^x\sin(x)\tanh(\sqrt{\Lambda}t)+\frac{\sqrt{\Lambda}\cos(x)v^t}{\cosh^2(\sqrt{\Lambda}t)}\right)\frac{\partial}{\partial v^x}
               \ ,\\  
                 \kappa^{\mathrm{ct}}_2 &= 
              -\cos(x)\frac{\partial}{\partial t} +\sqrt{\Lambda}\sin(x)\tanh(\sqrt{\Lambda} t) \frac{\partial}{\partial x} 
               \\&+v^x\sin(x)\frac{\partial}{\partial v^t}
                +\sqrt{\Lambda}\left(v^x\cos(x)\tanh(\sqrt{\Lambda}t)+\frac{\sqrt{\Lambda}\sin(x)v^t}{\cosh^2(\sqrt{\Lambda}t)}\right)\frac{\partial}{\partial v^x}
               \ ,\\  
                  \kappa^{\mathrm{ct}}_3 &= \frac{\partial}{\partial x} \ ,
\end{align*}

Using the previous results, it is easy to show that symmetric FP solutions have the form  $f=F(t,p,\ell)$, where 
\[
p=C^{\kappa_3}=\frac{\cosh^2(\sqrt{\Lambda}t)}{\Lambda} \,v^x\,,\quad \ell^2 =\delta_{ij}C^{\xi_i}C^{\xi_j}=\frac{(v^\theta)^2+\sin^2\theta (v^\varphi)^2}{\Lambda^2}\,.
\]
The reduced FP equation for $F$ is 
\begin{align*}
   & \sqrt{ m^2+\frac{\Lambda p^2}{\cosh^2(\sqrt{\Lambda}t)}+\Lambda \ell^2}\,\frac{\partial F}{\partial t}  \\
  & = \frac{\sigma^2}{2} 
     \Bigg( \left(\frac{1}{\Lambda} +\frac{\ell^2}{m^2}\right)  \partial^2_\ell +\left(\frac{1}{\Lambda \ell}+\frac{3\ell}{m^2}\right)\partial_\ell
     +\left(\frac{\cosh^2(\sqrt{\Lambda}t)}{\Lambda}+\frac{p^2}{m^2}\right)\partial^2_{p}+ \frac{2\ell p}{m^2} \partial_\ell\partial_p +\frac{3p}{m^2}\partial_{p}  \Bigg) F \,.
\end{align*}
\begin{remark}
    It is interesting to notice that in local coordinates, when the spacetime metric has the form
\[
ds^2=g_{tt}(t,r) dt^2+ g_{rr}(t,r) dr^2 +R^2(t,r) \left( d\theta^2+\sin^2\theta d\varphi^2\right) \ ,
\]
the vertical Laplacian acting on \textit{rotational-invariant} functions is just 
\[
\Delta^{\mathrm{ver}}_m  = \left( R^2 +\frac{\ell^2}{m^2}\right)  \partial^2_\ell +\left(\frac{R^2}{\ell}+\frac{3\ell}{m^2}\right)\partial_\ell
     +\left(\frac{1}{g_{rr}}+\frac{(v^r)^2}{m^2}\right)\partial^2_{v^r}+ \frac{2\ell v^r}{m^2}\partial_\ell\partial_{v^r} +\frac{3v^r}{m^2}\partial_{v^r}\,. 
\]
For example, for the Schwarzschild spacetime ($\Lambda = 0$) we recover \eqref{SchReducido} and it is also straightforward to write the corresponding operator in the case of Schwarzschild-de Sitter (or anti-de Sitter) spacetimes described by \eqref{schdsit}. 

Furthermore, an interesting open submanifold of  Nariai spacetime is the Bertotti-Kasner spacetime $(M_{\mathrm{BK}},g_{\mathrm{BK}})$, that is a particular example of the Kantowski-Sachs metrics. This spacetime was used by W. Rindler in   \cite{rindler1998birkhoff} to discuss Birkhoff's theorem (see also \cites{kasner1925algebraic,bertotti1959uniform}):
\begin{align*}
   M_{\mathrm{BK}} =\mathbb{R}^2\times \mathbb{S}^2\,,\quad   ds^2_{\mathrm{BK}} = - dt^2 + e^{2\sqrt{\Lambda}t} dr^2 + \frac{1}{\Lambda} (d\theta^2 + \sin^2{\theta} d\varphi^2) \ .
\end{align*}
The $\mathfrak{sl}(2)$ Killing fields are 
\[
 X_1 = \frac{\partial}{\partial r} \ ,\quad
               X_2= -\frac{1}{\sqrt{\Lambda}}\frac{\partial}{\partial t} +r\frac{\partial}{\partial r} \ ,\quad
                  X_3 = -2\sqrt{\Lambda}r\frac{\partial}{\partial t} +\left(e^{-2\sqrt{\Lambda}t}+\Lambda r^2\right)\frac{\partial}{\partial r} \ ,
            \]
and the reduced FP equation is
\begin{align}\label{BKreduced}
         &\sqrt{ m^2+e^{-2\sqrt{\Lambda}t}p^2+\Lambda \ell^2}\,\frac{\partial F}{\partial t}  
        \\ 
        \nonumber&=\frac{\sigma^2}{2} 
         \Bigg( \left(\frac{1}{\Lambda} +\frac{\ell^2}{m^2}\right)  \partial^2_\ell +\left(\frac{1}{\Lambda \ell}+\frac{3\ell}{m^2}\right)\partial_\ell+\left(e^{2\sqrt{\Lambda}t}+\frac{p^2}{m^2}\right)\partial^2_{p}+ \frac{2\ell p}{m^2} \partial_\ell\partial_p +\frac{3p}{m^2}\partial_{p}  \Bigg) F  \ ,
\end{align}
where, in this case,  $p=e^{2\sqrt{\Lambda}t} v^r$.

\end{remark}

\section{Conclusions and comments}{\label{sec:concl}}

In this paper, we have extensively discussed the relationship between the symmetry properties of a spacetime $(M,g)$ and relativistic diffusion processes on two relevant bundles over $M$: the (restricted) orthonormal frame bundle $\mathcal{SO}^+(M)$ and the observer bundle $U^mM$ of arbitrary positive mass $m$. We have demonstrated that the complete lift of a Killing vector field of $(M,g)$ commutes with both the fundamental and basic vector fields within $\mathcal{SO}^+(M)$. This observation indicates that the symmetries inherent to the spacetime $(M,g)$ also preserve the generator of the diffusion processes on both $\mathcal{SO}^+(M)$ and $U^mM$. Furthermore, we have identified the conditions that the symmetric solutions to the Fokker–Planck diffusion
equation satisfy on $\mathcal{SO}^+(M)$ and $U^mM$, and we have detailed how to establish the connection between them, thereby extending previous findings in the literature.

To connect our results with physical applications, we have demonstrated that, within the framework provided by $U^mM$, the dynamics of the particle system can be characterized by the one-particle distribution function $f$, which obeys the Fokker--Planck equation ensuring the conservation of the average of particle trajectories passing through a given Cauchy surface. This conservation principle can also be derived through the vanishing divergence of the current density vector field $\mathsf{J}$ associated with $f$. Additionally, we have proven that the divergence of the entropy current $\mathsf{S}$ is non-negative, and that the divergence of the energy-momentum tensor field $\mathsf{T}$ is proportional to the particle current density, consistent with the expected behavior of a diffusion process.

Finally, we have illustrated the results by using several physically relevant spacetimes, spanning both cosmological and astrophysical applications.

While all the results presented in this work can be established without resorting to symplectic techniques, we have chosen to follow this approach anticipating future uses. For instance, it streamlines the incorporation of additional physical information, as is the case of electromagnetic fields: by simply modifying the symplectic form $\mathrm{d}\alpha$ to $\mathrm{d}\alpha+ q F$, where $q$ represents the particle charge and $F$ denotes the Faraday tensor describing the electromagnetic field in spacetime, we can seamlessly integrate electromagnetic effects into our analysis. It may also be interesting to analyze recent generalizations of the Sasaki metric proposed in \cite{kapsabelis2024finsler} within the context of modified gravity.

\bibliographystyle{amsplain}

\end{document}